\documentclass[letterpaper,11pt]{article}
\usepackage[margin=0.95in]{geometry}
\usepackage{amsthm,amsmath}
\usepackage{amssymb}

\usepackage[ruled]{algorithm}
\usepackage[noend]{algpseudocode}

\usepackage{epsfig}

\renewcommand{\thefootnote}{\fnsymbol{footnote}}

\newtheorem{theorem}{Theorem}
\newtheorem{corollary}{Corollary}
\newtheorem{lemma}{Lemma}

\newtheorem{proposition}{Proposition}

\newtheoremstyle{redstyle}
     {3pt}
     {3pt}
     {\color{black}}
     {}
     {\color{red}\bfseries}
     {:}
     {.5em}
     {}
\theoremstyle{redstyle}

\newcommand{\NOGRAN}{\textsc{NoGran}}

\newcommand{\AlgD}{{\sc DiamUBr}}
\newcommand{\AlgN}{{\sc SizeUBr}}
\newcommand{\AlgG}{{\sc GranUBr}}

\newcommand{\AlgCRBE}{{\sc ChooseRepByEcho}}

\newcommand{\m}{\mathcal}
\newcommand{\cT}{{\mathcal T}}
\newcommand{\cN}{{\mathcal N}}
\newcommand{\cA}{{\mathcal A}}

\newcommand{\eps}{\varepsilon}

\newcommand{\gammapom}{\gamma}

\newcommand{\DIR}{\text{DIR}}

\newcommand{\tju}[1]{} 
\newcommand{\labell}[1]{\label{#1}} 
\renewcommand{\paragraph}[1]{\vspace*{0.5ex}\noindent {\bf #1}}
\newcommand{\remove}[1]{}

\newcommand{\dist}{\text{dist}}
\newcommand{\distM}{\text{distM}}
\newcommand{\side}{\text{side}}
\newcommand{\colour}{\text{color}}
\newcommand{\NAT}{{\mathbb N}}
\newcommand{\INT}{{\mathbb Z}}

\newcommand{\Transmit}{Inter-Box-Broadcast}

\newcommand{\Election}{Election}

\usepackage{color}

\newcommand{\dk}[1]{#1} 
\newcommand{\tj}[1]{#1} 
\newcommand{\comment}[1]{}

\newcommand{\sone}{asleep}
\newcommand{\stwo}{active}
\newcommand{\sthree}{idle}

\newcommand{\boxx}{\text{box}}

\title{Distributed Deterministic Broadcasting\\ in Wireless Networks of Weak Devices under the SINR Model\thanks{%
	This work was supported by the
	EPSRC grant EP/G023018/1.}
	}

\author{%
    Tomasz Jurdzinski\footnotemark[2]
    \and
    Dariusz R.~Kowalski\footnotemark[3]
    \and
    Grzegorz Stachowiak\footnotemark[2]
}

\begin{document}

\footnotetext[2]{Institute of Computer Science, University of Wroc{\l}aw, Poland.}

\footnotetext[3]{Department of Computer Science,
            University of Liverpool,
            Liverpool L69 3BX, UK.
            }


\date{}

\maketitle


\begin{abstract}
The Signal-to-Interference-and-Noise-Ratio model (SINR) is currently the most popular model for
analyzing communication in wireless networks. Roughly speaking, it allows receiving a message
if the strength of the signal carrying the message dominates over the
combined strength of the remaining signals and the background noise at the receiver.
There is a large volume of analysis done under the SINR model in the centralized setting,
when both network topology and communication tasks are provided as a part of the common input,
but surprisingly not much is known
in the ad hoc setting, when nodes have very limited knowledge about the network topology.
In particular, there is no theoretical study of deterministic solutions to multi-hop
communication tasks, i.e., tasks in which packets often have to be relayed in order to reach their destinations.
These kinds of problems, including broadcasting, routing, group communication, leader election,
and many others, are important from perspective of development of
future multi-hop wireless and mobile technologies, such as MANET, VANET, Internet of Things.

In this paper we initiate a study of distributed deterministic broadcasting
in ad-hoc wireless networks with uniform transmission powers under the SINR model.
We design algorithms in two settings: with and without local knowledge about immediate neighborhood.
In the former setting, our solution has almost optimal $O(D\log^2 n)$ time cost,
where $n$ is \tj{the size of a network}, $D$ is the eccentricity of the network and
$\{1,\ldots,N\}$ is the set of possible node IDs.
In the latter case, we prove an $\Omega(n\log N)$ lower bound and develop
an algorithm matching this formula, where $n$ is the number of network nodes.
As one of the conclusions, we derive that the inherited cost of broadcasting techniques in wireless networks
is much smaller, by factor around $\min\{n/D,\Delta\}$, than the cost of learning the immediate neighborhood.
\tj{Finally, we develop a $O(D\Delta\log^2 N)$ algorithm for the setting without local knowledge,
where $\Delta$ is the upper bound on the degree of the communication graph of a network. This algorithm
is close to a lower bound $\Omega(D\Delta)$.}

\tju{przen. z our results}
\tj{In the model without local knowledge, we 
take advantage
of the fact that efficient deterministic distributed communication is possible (in the SINR model)
between stations which are very close, despite large amount
of interferences caused by other transmitters.
This feature somehow compensates inconveniences caused by distant
interferences and makes it possible to obtain a broadcasting algorithm with
efficiency similar to that obtained for UDG radio networks. However,
unlike in the UDG radio networks model, the (lower) bounds apply also
for randomized solutions. In other words, randomization
does not substantially help in ad hoc distributed broadcasting in a large class of
networks.
}


\end{abstract}

\thispagestyle{empty}
\setcounter{page}{0}


\renewcommand{\thefootnote}{\arabic{footnote}}

\newpage



\setcounter{section}{0}
\setcounter{proposition}{0}
\setcounter{lemma}{0}
\setcounter{theorem}{0}
\setcounter{corollary}{0}
\setcounter{algorithm}{0}
\setcounter{equation}{0}
\setcounter{figure}{0}

\section{Introduction}

In this work we consider a broadcasting problem in ad-hoc wireless networks under the
Signal-to-Interference-and-Noise-Ratio model (SINR).
Wireless network consists of $n$ stations, also called nodes, with unique integer IDs in the range $\{1,\ldots,N\}$
and uniform transmission powers, deployed in the two-dimensional space with Euclidean metric.
Each station initially knows only its own ID and location, parameters $n$ and $N$.
A communication (or reachability) graph of the network is the graph defined on network nodes
and containing links $(v,w)$ such that if $v$ is the only transmitter in the network then $w$ receives
the message transmitted by $v$.
We consider two settings: one with local knowledge, in which each station knows also its neighbors
(i.e., stations reachable by a direct transmission), and the other when no extra knowledge is assumed.

In the broadcasting problem, there is one designated node, called the source, which has a piece
of information (called a source message or a broadcast message)
that must be delivered to all other accessible nodes by using wireless communication.
In the beginning, only the source is active from perspective of the broadcast task,
and other nodes join the execution after receiving the broadcast message for the first time.
The goal is to minimize the worst-case time for accomplishing the broadcasting task.



\subsection{Previous and Related Results}

Recent development of deterministic protocols for wireless communication, e.g.,
CDMA-based technologies, and rapidly growing scale of ad hoc wireless networks,
poses new challenges for design of efficient deterministic distributed protocols.
In this work, we study
the problem of {\em distributed deterministic broadcasting} in ad hoc wireless networks,
which,  to the best of our knowledge,
has not been theoretically studied under the SINR model, from perspective of worst-case complexity.
SINR model is currently considered the most adequate
among the models of wireless networks.
Furthermore,
no other communication task involving multi-hop message propagation has been theoretically studied
from perspective of distributed deterministic solutions in the SINR setting.
In what follows, we list most relevant results in the SINR model, and
the state of the art obtained in the older Radio Network model.

\paragraph{SINR model.}
In the 
SINR model in ad hoc setting,
deterministic {\em local} broadcasting, in which nodes have to inform
only their neighbors in the corresponding reachability graph,
was studied in \cite{YuWHL11}.
The considered setting allowed power control by algorithms,
in which, in order to avoid collisions,
stations could transmit with any power smaller than the
maximal one.
Randomized solutions for contention resolution~\cite{KV10}
and local broadcasting~\cite{GoussevskaiaMW08} were also obtained.

There is a vast amount of work on centralized algorithms under the SINR model.
The most studied problems include connectivity, capacity maximization,
link scheduling types of problems (e.g.,\ \cite{FanghanelKRV09,Kesselheim11,AvinLPP09}).
For recent results and references we refer the reader to the survey~\cite{WatSurv}.
Multiple Access Channel properties were also recently studied
under the SINR model, c.f.,~\cite{RichaSSZ}.

\paragraph{Radio network model.}
There are several papers analyzing deterministic broadcasting in the radio model of wireless networks,
under which a message is successfully heard if there are no other simultaneous transmissions
from the {\em neighbors} of the receiver in the communication graph.
This model does not take into account the real strength of the received signals, and also the signals
from outside of some close proximity.
In the geometric ad hoc setting, Dessmark and Pelc~\cite{DessmarkP07} were the first who studied
this problem. They analyzed the impact of local knowledge, defined as a range within which
stations can discover the nearby stations.
Unlike most research on broadcasting problem and the assumptions of this paper,
Dessmark et\ al. \cite{DessmarkP07} assume spontaneous wake-up of stations. That is,
stations are allowed to do some pre-processing
(including sending/receiving messages) prior receiving
the broadcast message for the first time.
Moreover it is assumed in \cite{DessmarkP07} that IDs are from $\{1,\ldots,n\}$,
which makes the setting even
less comparable with the one considered in this work.
Emek et al.~\cite{EmekGKPPS09} designed a broadcast algorithm
working in time $O(Dg)$
in UDG radio networks with eccentricity $D$ and granularity $g$, where
eccentricity was defined as the minimum number of hops to propagate the broadcast message throughout
the whole network and
granularity was defined as
the
inverse of the minimum distance between any two stations.
Later, Emek et al.~\cite{EmekKP08} developed a matching lower bound $\Omega(Dg)$.
%
There were several works analyzing deterministic broadcasting in geometric graphs in the centralized radio setting,
c.f.,~\cite{GasieniecKKPS08,GasieniecKLW08,SenH96}.
%

The problem of broadcasting is well-studied in the setting of graph radio model, in which stations
are not necessarily deployed in a metric space;
here we restrict to only the most relevant results.
In deterministic ad hoc setting with no local knowledge, the fastest $O(n\log(n/D))$-time algorithm in symmetric networks was developed by Kowalski~\cite{Kow-PODC-05}, and almost matching lower
bound was given by Kowalski and Pelc~\cite{KP-DC-05}.
%
For recent results and references in less related settings we refer the reader
to~\cite{DeMarco-SICOMP-10, KP-DC-07, CzumajRytter-FOCS-03,Censor-HillelGKLN11,GalcikGL09}

There is vast literature on randomized algorithms for broadcasting in graph radio model.
Since they are quite efficient, there are very few studies of the problem restricted to
geometric setting. However, when mobility of stations is assumed, location and movement
of stations on the plane is natural. Such settings were studied e.g.,\ in
\cite{Farach-ColtonAMMZ11,Farach-ColtonM07}.


%
%


%

\subsection{Our Results}
In this paper we present the first study on deterministic broadcasting in wireless connected networks
deployed in two dimensional Euclidean space under the SINR model.
We distinguish between the two settings:
with and without local knowledge about neighbors in the communication graph.
In the former model, we developed a broadcasting
algorithm with time complexity $O(n\log N)$, which matches the lower bound \tj{(Section~\ref{s:anonymous})}.
\tj{Then, an algorithm finishing broadcasting in time $O(D\Delta\log^2N)$ is presented, where $\Delta$ is the
largest degree of a vertex in the reachability graph (Section~\ref{s:algdeg}).
This algorithm is close to the lower bound $\Omega(D\Delta)$ -- see Section~\ref{s:lower}.}
Our solution for networks with local knowledge works in time
$O(D\log^2 n)$, which provides $O(\log^2 n)$ overhead over the straightforward $\Omega(D)$ lower bound,
and is faster than the algorithms for anonymous networks in every network with eccentricity
$D=o(n/\log N)$ \tj{or maximal degree $\Delta=\omega(1)$.
It also implies that the cost of
learning
neighborhoods by stations in wireless network is
much higher, by factor around $n/D$ \tj{or $\Delta$}, than the cost of broadcast itself (performed when such neighborhoods
are provided).
Importantly, the algorithm for networks with local knowledge works for any path loss parameter
$\alpha\geq 2$ (though additional multiplicative $\log^2 N$ factor
appears in complexities of algorithms for $\alpha=2$), while the algorithms without local knowledge are applicable only when $\alpha>2$.
}

Our results rely on novel techniques which simultaneously
exploit specific properties of conflict resolution in the
SINR model (see e.g. \cite{AvinEKLPR09}) and algorithmic
techniques developed for radio networks model. In particular,
in the model with local knowledge, we show how to efficiently combine a novel
SINR-based leader election technique, ensuring several parallel communications
inside range area of one station (which is unfeasible to achieve in
radio networks model), with the approach simulating collision
detection in radio networks (c.f.\ \cite{KP04}).
As a result,
we develop a general transformation of algorithms relying on the knowledge
of network granularity \tj{(Section~\ref{s:granularity-unknown})into algorithm of asymptotically 
similar performance
that do not require such knowledge.}

%
\tj{In the model without local knowledge, we 
take advantage
of the fact that efficient deterministic distributed communication is possible (in the SINR model)
between stations which are very close, despite large amount
of interferences caused by other transmitters.
This feature somehow compensates inconveniences caused by distant
interferences and makes possible to achieve broadcasting algorithm with
efficiency similar to that obtained for UDG radio networks. However,
unlike in the UDG radio networks model, the (lower) bounds apply also
for randomized solutions. In other words, randomization
does not substantially help in ad hoc distributed broadcasting in a large class of
networks.
}
%



%

\vspace*{-1ex}
\section{Model, Notation and Technical Preliminaries}

Throughout the paper, $\NAT$ denotes the set of natural numbers,
$\NAT_+$ denotes the set $\NAT\setminus\{0\}$, and $\INT$
denotes the set of integers.
For $i,j\in\INT$, we use the notation $[i,j]=\{k\in\NAT\,|\,i\leq k\leq j\}$
and $[i]=[1,i]$.
%

We consider a wireless network consisting of $n$ {\em stations}, also called {\em nodes},
deployed into a two dimensional
Euclidean space and communicating by a wireless medium.
%
%
All stations have unique integer IDs in set $[N]$. 
Stations of a network are denoted by letters $u, v, w$, which simultaneously
denote their IDs.
%
%
%
Stations are located on the plane with {\em Euclidean metric} $\dist(\cdot,\cdot)$,
and each station knows its coordinates.
%
%
Each station $v$ has its {\em transmission power} $P_v$, which is a positive real number.
There are three fixed model parameters: path loss
\tj{$\alpha\geq 2$,}
threshold $\beta\ge 1$, and ambient noise $\cN\ge 1$.
The $SINR(v,u,\cT)$ ratio, for given stations $u,v$ and a set of (transmitting) stations $\cT$,
is defined as follows:
\vspace*{-1ex}
\begin{equation}\label{e:sinr}
SINR(v,u,\cT)
=
\frac{P_v\dist(v,u)^{-\alpha}}{\cN+\sum_{w\in\cT\setminus\{v\}}P_w\dist(w,u)^{-\alpha}}
\end{equation}
In the {\em Signal-to-Interference-and-Noise-Ratio model} (SINR) considered in this work,
station $u$ successfully receives a message from station $v$ in a round if
$v\in \cT$, $u\notin \cT$, and:
\begin{itemize}
\vspace*{-1ex}
\item
$SINR(v,u,\cT)\ge\beta$, where $\cT$ is the set of stations transmitting at that time, and
\vspace*{-1ex}
\item
$P_v\dist^{-\alpha}(v,u)\geq (1+\eps)\beta\cN$,
\end{itemize}
\vspace*{-1ex}
where $\eps>0$ is a fixed {\em sensitivity parameter} of the model.
\tj{
The above definition is common in the literature, c.f.,~\cite{KV10}.\footnote{%
The first condition is a straightforward application of the SINR ratio,
comparing strength of one of the received signals with the remainder.
The second condition enforces the signal to be sufficiently strong in order to be
distinguished from the background noise, and thus to be decoded.
\tj{Moreover, this condition ensures that all transmission powers are high enough so that some interference can be tolerated.}
}
} 

In the paper, we assume for the sake of
clarity of presentation that $\beta=1$ and
$\cN=1$.
These assumptions can be dropped without harming the asymptotic performances of
the presented algorithms and lower bounds formulas.

\paragraph{Ranges and uniformity.}
The {\em communication range} $r_v$ of a station $v$ is the radius of the circle in which a message transmitted
by the station is heard, provided no other station transmits at the same time.
A network
is
{\em uniform}, when ranges (and thus transmission powers) of all stations are equal,
or {\em nonuniform} otherwise.
In this paper, only uniform networks are considered.
For clarity of presentation
we make the assumption that all powers are equal to $1$, i.e., $P_v=1$ for each $v$.
The assumption that the values of $P_v$ are $1$ can be dropped without changing
asymptotic formulas for presented algorithms and lower bounds.
Under these assumptions, $r_v=r=(1+\eps)^{-1/\alpha}$
for each station $v$.
%
%
%
%
The {\em range area} of a station with range $r$
located at the point $(x,y)$ is defined as the circle with radius $r$.

\paragraph{Communication graph and graph notation.}
The {\em communication graph} $G(V,E)$, also called the {\em reachability graph}, of a given network
consists of all network nodes and edges $(v,u)$ such that $u$ is in the range area of $v$.
%
Note that the communication graph is symmetric for uniform networks, which are considered
in this paper.
By a {\em neighborhood} of a node $u$ we mean the set (and positions) of all 
neighbors of $u$, i.e., the set $\{w\,|\, (w,u)\in E\}$ in the communication graph $G(V,E)$
of the underlying network.
The {\em graph distance} from $v$ to $w$ is equal to the length of a shortest path from $v$ to $w$
in the communication graph, where the length of a path is equal to the number of its edges.
The {\em eccentricity} of a node
is the maximum graph
distance from this node to all other nodes
(note that the eccentricity is of the order of the diameter if the communication
graph is symmetric --- this is also the case in this work).
%

We say that a station $v$ transmits {\em $c$-successfully} in a round
$t$ if $v$ transmits a message in round $t$ and this message is heard by
each station $u$ in distance smaller or equal to $c$ from $v$. We say that a station
$v$ transmits {\em successfully} in round $t$ if it transmits $r$-successfully, i.e.,
each of its neighbors in the communication graph can hear its message.
Finally, $v$ transmits {\em successfully} to $u$ in round $t$ if $v$ transmits
a message in round $t$ and $u$ receives this message.

\paragraph{Synchronization.}
It is assumed that algorithms work synchronously in rounds, each station can
either act as a sender or as a receiver during a round.
We do not assume global clock ticking -- as it can be coordinated by updating
round counter and passing it along the network with messages.

\paragraph{Collision detection.}
We consider the model without {\em collision detection}, that is,
if a station $u$ does not receive a message in a round $t$, it has no information
whether any other station was transmitting
in that round
and about the value of $SINR(v,u,\m{T})$, for any station $u$, where $\m{T}$ is the set
of transmitting stations in round $t$.

\paragraph{Broadcasting problem and complexity parameters.}
In the broadcasting problem studied in this work, there is one distinguished node, called the {\em source},
which initially holds a piece of information (also called a source message or a broadcast message).
The goal is to disseminate this message to all other nodes by sending messages along the network.
The complexity measure is the worst-case time to accomplish the broadcast task,
taken over all connected networks with specified parameters.
Time, also called the {\em round complexity}, denotes here the number of communication rounds in
the execution of a protocol: from the round when the source is activated
with its broadcast message till the broadcast task is accomplished (and each station is aware of this fact).
For the sake of complexity formulas, we consider the following parameters:
$n$, $N$, $D$, and $g$, where:
$n$ is the number of nodes,
$[N]$ is the range of IDs,
$D$ is the eccentricity of the source,
and $g$ is the granularity of the network, defined as $r$ times the
inverse of the minimum distance between any two stations (c.f.,~\cite{EmekGKPPS09})
divided by $r$.

\paragraph{Messages and initialization of stations other than source.}
We assume that a single message sent in the execution of any algorithm
can carry the broadcast message and at most polynomial, in the
size of the network, number of control bits in the size of the network.
For simplicity of analysis, we assume that every message sent during the execution
of our broadcast protocols contains the broadcast message; in practice, further optimization
of a message content could be done in order to reduce the total number of transmitted bits in real executions.
A station other than the source starts executing the broadcasting protocol
after the first successful receipt of the broadcast message; we call it
a {\em non-spontaneous wake-up model}, to distinguish from other possible settings,
not considered in this work,
where stations could be allowed to do some pre-processing
(including sending/receiving messages) prior receiving
the broadcast message for the first time.
We say that a station that received the broadcast message is {\em informed}.

\paragraph{Knowledge of stations.}
Each station knows its own ID, location, and parameters $n$, $N$.
Some subroutines use the granularity $g$ as a parameter, though
our main algorithms can use these subroutines without being aware
of the actual granularity of the input network.
We distinguish between {\em ad hoc} networks, where stations do not know anything
about the topology of the network at the beginning of the execution of an algorithm,
and networks with {\em local knowledge},
in which each station knows locations and IDs of its neighbors in the communication graph.

\remove{
depending on the algorithm, other general network parameters such as:
diameter $D$ of the imposed communication graph,
or granularity of the network $g$, defined as the inverse of the smallest
distance between any pair of stations.\footnote{%
In many cases, the assumption about the knowledge of $D,g$ can be dropped,
by running parallel threads for different ranges of values of these parameters and implementing an additional coordination mechanism between the threads.}
}

\subsection{Grids and Schedules}

Given a parameter $c>0$, we define 
a partition of the $2$-dimensional space
into square boxes of size $c\times c$ by the grid $G_c$, in such a way that:
all boxes are aligned with the coordinate axes,
point $(0,0)$ is a grid point,
each box includes its left side without the top
endpoint and its bottom side without the right endpoint and
does not include its right and top sides.
We say that $(i,j)$ are the coordinates
of the box with its bottom left corner located at $(c\cdot i, c\cdot j)$,
for $i,j\in \INT$. A box with coordinates
$(i,j)\in\INT^2$ is denoted $C(i,j)$.
As observed in \cite{DessmarkP07,EmekGKPPS09}, the {\em grid} $G_{r/\sqrt{2}}$
is very useful in design of algorithms for geometric radio networks, provided
$r$ is equal to the range of each station.
This follows from the
fact that $r/\sqrt{2}$ is the largest parameter of a grid such that each
station in a box is
in the range of every other station in that box.
In the following, we fix $\gamma=r/\sqrt{2}$, where $r=(1+\eps)^{-1/\alpha}$, and call $G_{\gamma}$
the {\em pivotal grid}. If not stated otherwise, our considerations will
refer to (boxes of) $G_{\gamma}$.

Two boxes $C,C'$ are {\em neighbors} in a network if there are
stations $v\in C$ and $v'\in C'$ such that edge $(v,v')$ belongs to the
communication graph of the network. Boxes $C(i,j)$ and $C'(i',j')$ are {\em adjacent} if
$|i-i'|\leq 1$ and $|j-j'|\leq 1$ (see Figure~\ref{fig:adjacent}).
For a station $v$ located in position $(x,y)$ on the plane we define its {\em grid
coordinates} with respect to the grid $G_c$ as the pair of integers $(i,j)$ such that the point $(x,y)$ is located
in the box $C(i,j)$ of the grid $G_c$ (i.e., $ic\leq x< (i+1)c$ and
$jc\leq y<(j+1)c$).
If not stated otherwise, we will refer to grid coordinates with respect
to the pivotal grid.

A (general) {\em broadcast schedule} $\mathcal{S}$ of length $T$
wrt  $N\in\NAT$ is a mapping
from $[N]$ to binary sequences of length $T$.
A station
with identifier $v\in[N]$ {\em follows}
the schedule $\m{S}$ of length $T$ in a fixed period of time consisting of $T$ rounds,
when
$v$ transmits a message in round $t$ of that period iff
the 
position $t\mod T$ of
$\m{S}(v)$ is equal to $1$.

A {\em geometric broadcast schedule} $\mathcal{S}$ of length $T$
with parameters $N,\delta\in\NAT$, $(N,\delta)$-gbs for short, is a mapping
from $[N]\times [0,\delta-1]^2$ to binary sequences of length $T$.
Let $v\in[N]$ be a station whose grid coordinates
with respect to
the grid $G_c$ are equal to $(i,j)$.
We say that $v$ {\em follows}
$(N,\delta)$-gbs $\m{S}$ 
for the grid $G_c$
in a fixed period of time, 
when $v$ transmits a message in round $t$ of that period iff
the $t$th position of
$\m{S}(v,i\mod \delta,j\mod\delta)$ is equal to $1$.
A set of stations $A$ on the plane is {\em $\delta$-diluted} wrt $G_c$, for $\delta\in\NAT\setminus\{0\}$, if
for any two stations $v_1,v_2\in A$ with grid coordinates $(i_1,j_1)$ and $(i_2,j_2)$, respectively,
the relationships $(|i_1-i_2|\mod \delta)=0$ and $(|j_1-j_2|\mod \delta)=0$ hold.

Let $\m{S}$ be a general broadcast schedule wrt $N$ of length $T$,
let $c>0$ and $\delta>0$, $\delta\in\NAT$.
A $\delta$-dilution of a  $\m{S}$ 
is defined as a $(N,\delta)$-gbs $\m{S}'$ such that the bit $(t-1)\delta^2+a\delta+b$
of $\m{S}'(v,a,b)$ is equal to $1$ iff the bit $t$ of $\m{S}(v)$
is equal to $1$. That is, each round $t$ of $\m{S}$ is
partitioned
into $\delta^2$ rounds
of $\m{S}'$, indexed by pairs $(a,b)\in [0,\delta-1]^2$, such that a station
with grid coordinates $(i,j)$ in $G_c$ is
allowed to send messages only in rounds with index $(i\mod\delta,j\mod\delta)$,
provided schedule $\m{S}$ admits a transmission in its (original) round $t$.
Since we will usually apply dilution to the pivotal grid, it is assumed that all references
to a dilution concern that grid, unless stated otherwise.

\begin{figure}
\begin{center}
\epsfig{file=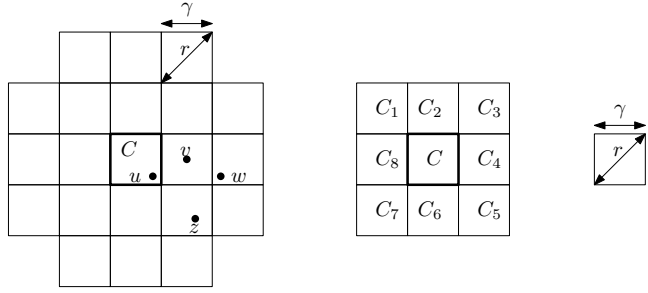, scale=0.8}
\end{center}
\caption{If $v,w,z$ are in the range are of $u$, then boxes
containing $v,w,$ and $z$ are neighbors of $C$. The first figure
contains all $20$ boxes which can be neighbors of $C$.
The boxes
$C_1,\ldots,C_8$ are adjacent to $C$.}
\label{fig:adjacent}
\end{figure}%
Observe that, since ranges of stations are equal to the length
of diagonal of boxes of the pivotal grid, a box $C(i,j)$ can have at most
$20$ neighbors (see Figure~\ref{fig:adjacent}).
We define the set $\DIR\subset[-2,2]^2$ such  that $(d_1,d_2)\in\DIR$ iff
it is possible that boxes with coordinates $(i,j)$ and $(i+d_1,j+d_2)$
can be neighbors.
Given $(i,j)\in\INT^2$ and $(d_1,d_2)\in\DIR$, we say that the box $C(i+d_1,j+d_2)$
is {\em located in direction} $(d_1,d_2)$ from the box $C(i,j)$.

\section{Algorithms for Networks with Local Knowledge}
\labell{s:local}

In this section we describe our broadcasting algorithms
for networks with local knowledge, i.e., under the assumption
that each stations knows (IDs and locations) of all stations
in its range area. Recall that we also assume that stations
know $n$, the size of the network and $N$, the range of identifiers.
We start with presenting a generic algorithmic scheme and tools for analysis.
Next, we describe an algorithm for networks with additionally known granularity
bound $g$, i.e., parameters $n,N$ and $g$ are known to the stations in the beginning
of the execution. Complexity of this algorithm is expressed in terms of $D$ and $g$;
note however that stations do not need any information about $D$ in order to execute
our algorithms.
Finally, using this algorithm as a subroutine, we provide a
solution for the general setting when only $n$ and $N$ are known.

\subsection{Generic Algorithmic Scheme}
\labell{s:generic}

\tj{In the first step of each broadcasting algorithm, the source sends the broadcast message.
Then, our broadcasting algorithms repeat several times the procedure
{\Transmit}, whose $i$th repetition is aimed to transmit
the broadcast message from boxes of the pivotal grid containing
at least one station that has 
received the broadcast message in the previous execution of {\Transmit} (or from the source)
to boxes which are their neighbors.}
%

Each station $v$ of the network is in state $s(v)$, which may be equal to
one of the following three values: \sone, \stwo, or \sthree.
At the beginning of execution of each of our broadcasting algorithms, the source
sends the broadcast message and all stations in its box of the pivotal
grid set their states to \stwo, while all the remaining stations are in the {\sone} state.
The states of stations change only at the end of {\Transmit},
according to the following rules:
\begin{itemize}
\vspace*{-1ex}
\item
All stations in state {\stwo} change their state to {\sthree}.
\item
A station $u$ changes its state from {\sone} to {\stwo} if it has received
the broadcast message from a station $v$ in the current
execution of {\Transmit}
such that either $v$ was in state {\stwo} (at the beginning of the current
execution of \Transmit) or $v$ belongs to the same box of the pivotal grid
as $u$. \tj{That is, let $C$ be a box of the pivotal grid,
let $u\in C$ be in state {\sone} at the beginning of {\Transmit}.
The only possibility that $u$ receives a message and it does not
change its state from {\sone} to {\stwo} at the end of {\Transmit} is that
each message received by $u$ is sent by a station $v$ which is in state {\sone} 
when it sends the message and $v\not\in C$.}
\end{itemize}
Our goal is to preserve the following invariant during the execution of our algorithms:
\begin{enumerate}
\item[(I)]\labell{i:I}
For each box $C$ of the pivotal grid, states of all stations located inside $C$
are equal.
\end{enumerate}
%
The intended property of an execution of {\Transmit} is:
\begin{enumerate}
\item[(P)]\labell{i:P}
The broadcast
message is (successfully)
sent from each box $C$ containing stations in state $\stwo$ to all stations located
in boxes which are neighbors of $C$.
(Recall that a box $C'$ is a neighbor of a box $C$ if there are
stations $v\in C$ and $v'\in C'$ such that edge $(v,v')$ belongs to the
communication graph.)
\end{enumerate}
Note that, since stations move to the state $\stwo$ only after receiving the
broadcast message, the following fact holds.
\begin{proposition}\labell{prop:invariants}
If (I) and (P) are satisfied, the source message is transmitted
to the whole network in time $O(D\cdot T(n))$, where $T(n)$ is time complexity
of one execution of {\Transmit}.
\end{proposition}

In what follows, we give a specification of {\Transmit} first under the
assumption of known granularity $g$, and later we remove that assumption.

\vspace*{-1ex}
\subsection{A Granularity-Dependent Algorithm}
\labell{s:granularity-unknown}

In this section we describe a broadcasting algorithm
whose complexity depends on granularity. We assume that
granularity $g$ is known to all stations of the network.
First, we present a general leader election algorithm,
which, given a set of stations $V$ with granularity $g$, elects a leader
in each box of the pivotal grid containing at least one
element of $V$, in time $O(\log g)$. Then, using this algorithm, we describe
how to implement {\Transmit} in time $O(\log g)$ in such a way that (I) and (P) are
preserved.


\subsubsection{Leader Election}
Let $I_1=[i_1,j_1)$, $I_2=[i_2,j_2)$ be segments on a
line, whose endpoints belong to the grid
$G_x$.
%
\tj{The {\em box-distance} between $I_1$ and $I_2$ with respect to $G_x$ is zero when $I_1\cap I_2\neq\emptyset$,
and it is equal to $\min(|i_1-j_2|/x, |i_2-j_1|/x)$ otherwise. Given two rectangles
$R_1$, $R_2$, whose vertices belong to $G_x$, the box-distance $\distM(R_1,R_2)$ between $R_1$
and $R_2$ is equal to the maximum of the box-distances between projections
of $R_1$ and $R_2$ on the axes defining the first and the second dimension in the Euclidean
space.}

We say that a function $d_{\alpha}:\NAT\to\NAT$ is {\em flat} for $\alpha\geq 2$ if
\begin{equation}
d_{\alpha}(n)=\left\{
\begin{array}{rcl}
O(1) & \mbox{ for } & \alpha>2\\
O(\log n) & \mbox{ for } & \alpha=2
\end{array}
\right.
\end{equation}

\begin{lemma}\labell{lm:leaderGranularity}
Given a set of stations $V$ with granularity $g$, one can choose the leader in each
box of the pivotal grid containing at least one element of $V$ in $O(d_{\alpha}^2(n)\log g)$ rounds,
where $d_{\alpha}(n)$ is a flat function.

Moreover, if polynomial size of messages is allowed, each station can learn
positions of all (active) stations located in its box in $O(\log g)$ rounds.
\end{lemma}

The remaining part of this section is devoted to the proof of Lemma~\ref{lm:leaderGranularity}.
\begin{proposition}\labell{prop:lead1}
\tj{For each $\alpha\geq 2$ 
and $\eps>0$,
there exists a flat function $d_{\alpha}(n)$
such that
the following properties hold.}
Assume that a set of $n$ stations $A$ is $d$-diluted wrt the grid $G_x$, where $x=\gamma/c$, $c\in\NAT$, $c>1$
and $d\geq d_{\alpha}(n)$.
Moreover, at most one station from $A$ is located in each box of $G_x$. Then,
if all stations from $A$ transmit simultaneously, each of them
is $\frac{2r}{c}$-successful. Thus, in particular, each station from
a box $C$ of $G_x$ can transmit its message
to all its neighbors located in $C$ and in boxes $C'$ of $G_x$ which are adjacent to $C$.
\end{proposition}
\begin{proof}
Recall that $r=(1+\eps)^{-1/\alpha}$ and $\gamma=r/\sqrt{2}$.
First, assume that $\alpha>2$.
Consider any station $u$ in distance smaller or equal to $\frac{2r}{c}\leq 2\sqrt{2}x<3x$ to a station $v\in A$.
Then, the signal from $v$ received by $u$ is at least
$$\frac1{\left(\frac{2r}{c}\right)^{\alpha}}=\left(\frac{c}{2r}\right)^{\alpha}.$$
Now, we would like to derive an upper bound on interferences caused by stations in $A\setminus\{v\}$
at $u$.
Let $C$ be a box of $G_x$ which contains $v$.
The fact that $A$ is $d$-diluted
wrt $G_x$ implies that the number of boxes containing elements of $A$
which are in box-distance
$id$ from $C$ is at most $8(i+1)$ (see Figure~\ref{fig:gran}).
Moreover, no box in distance $j$ from $C$ such
that ($j\mod d\neq 0$) contains elements of $A$.
\begin{figure}
\begin{center}
\epsfig{file=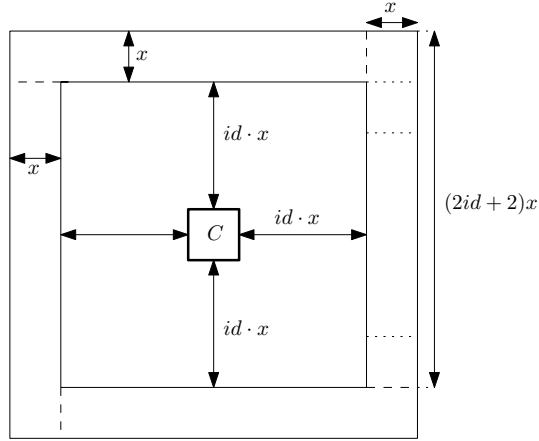, scale=0.8}
\end{center}
\caption{Boxes in distance $id$ from $C$ form a frame partitioned into four
rectangles of size $x\times (2id+2)x$. Each of these rectangles contain at most $i+1$
boxes such that any two of them are in box-distance at least $d$.}
\label{fig:gran}
\end{figure}%
Finally, for a station $v\in C$ and a station $w\in C'$ such that $\distM(C,C')=j$,
the inequality $\dist(v,u)\geq jx$ is satisfied.
Note that our goal is {\em not} to evaluate interferences at $v\in C$, but at any station
$u$ such that $\dist(u,v)\leq \frac{2r}c<3x$. Therefore, $u\in C'$ such that $\distM(C,C')<3$, where
$C'$ is a box of $G_x$.
For a fixed $d>3$, the total noise and interferences $I$ caused by
all elements of $A\setminus\{v\}$ at $u$ is at most
$$\cN+\sum_{i=1}^{n}8(i+1)\cdot\frac{1}{(i\bar{d}x)^{\alpha}}$$
where $d\geq\bar{d}\geq d-3$,
since there are at most $8(i+1)$ nonempty boxes in box-distance $i\cdot d$
from the box $C$
in $d$-diluted instance and the box-distance between $C$ and the box $C'$ containing $u$
is at most $2$. Furthermore,
$$I\leq 1+ 8\cdot\left(\frac{1}{\bar{d}x}\right)^{\alpha}\cdot\sum_{i=0}^{n}(i+1)^{1-\alpha}\leq 1+8\left(\frac{c\sqrt{2}}{r \bar{d}}\right)^{\alpha}\sum_{i=1}^{n}i^{1-\alpha}=1+8d_{\alpha}(n)\left(\frac{\sqrt{2}c}{r\bar{d}}\right)^{\alpha}$$
where $d_{\alpha}(n)=\sum_{i=1}^{n}i^{1-\alpha}=1+\zeta(\alpha-1)$, $\zeta$ is the Riemann zeta function and $\cN=1$.
So, 
the signal from $v$ is received at $u$ if the following
inequality is satisfied
\begin{equation}\label{eq:signal}
1+8d_{\alpha}(n)\left(\frac{\sqrt{2}c}{r\bar{d}}\right)^{\alpha}\leq \left(\frac{c}{2r}\right)^{\alpha}
\end{equation}
which is equivalent to
$$\bar{d}\geq 2\sqrt{2}\left(\frac{8d_{\alpha}(n)}{1-(2r/c)^{\alpha}}\right)^{1/\alpha}.$$
Assuming that $c\geq 2$, we have $1-(\frac{2r}{c})^{\alpha}\geq 1-r^{\alpha}$
and therefore (\ref{eq:signal}) is satisfied for each
$\bar{d}\geq 2\sqrt{2}\left(\frac{8}{1-r^{\alpha}}\right)^{1/\alpha}d_{\alpha}(n)$
or $d\geq 3+2\sqrt{2}\left(\frac{8}{1-r^{\alpha}}\right)^{1/\alpha}d_{\alpha}(n)$.


\paragraph{Note on dependence on $\eps$:}
by substituting $r:=(1+\eps)^{-1/\alpha}$, one can check that $d=O((1/\eps)^{1/\alpha})$ for $\alpha>2$.

\end{proof}

The following corollary is a straightforward application of Proposition~\ref{prop:lead1} for $c=2$.
\begin{corollary}\labell{cor:dilsuc}
For each $\alpha\geq2$ 
there exists
a flat function $d_{\alpha}:\NAT\to\NAT$
such that the following
property is satisfied: \\
Let $A$ be a set of $O(n)$ stations  on the plane which is
$\delta$-diluted wrt the pivotal grid $G_{\gamma}$,
where $\delta\geq d_{\alpha}(n)$ and each box
contains at most one element of $A$. Then, if all elements of $A$ transmit
messages simultaneously in the same round $t$ and no other station is transmitting
a message in $t$, each of them transmits successfully.\\
\end{corollary}



We say that a box $C$ of the grid $G_x$ has the {\em leader} from set $A$ if there is one station
$v\in A$ located in $C$
with status {\em leader} and all stations from $A$ located in $C$ know which station it is.

\begin{proposition}\labell{prop:lead2}
Assume that $A$ is a set of leaders in some boxes of the grid $G_x$, $x\leq\frac{\gamma}{2}$, and
each station knows whether it belongs to $A$. Then,
it is possible to choose the leader of each box of $G_{2x}$ containing at least one element of $A$
in $O(d_{\alpha}(n))$ rounds.
\end{proposition}
\begin{proof}
Note that each cell of $G_{2x}$ consists of four boxes of $G_x$. Let us fix some labeling of this four
boxes by the numbers $\{1,2,3,4\}$, the same in each box of $G_{2x}$.
Now, assign to each 
station from $A$
the label $l\in[1,4]$ corresponding
to its position it the box of $G_{2x}$ containing it. We ``elect'' leaders in $G_{2x}$ in four phases
$F_1,\ldots,F_4$. Phase $F_i$ is just the application of Proposition~\ref{prop:lead1} for $A$ equal to the
set of leaders with label $i$.
%
\tj{That is, we first have a general broadcast schedule $S$ of length $4$ such that
position $i$ of $S(v)$ is equal to $1$ iff label of $v$ is $1$. Then, $S$ is
$d$-diluted wrt $(N,x)$, where $d\gets d_{\alpha}(n)$ and $d_{\alpha}$ is the
constant from Proposition~\ref{prop:lead1}.}
Therefore, each leader from $A$ can hear messages of all other (at most)
three leaders located in the same box of $G_{2x}$. Then, for a box $C$ of $G_{2x}$, the leader with the
smallest label (if any) among leaders of the four sub-boxes of $C$ becomes the leader of $C$.
\end{proof}

Assume that granularity of a network is 
equal to $g$. Let $h=\min_{i\in\NAT}({2^i\,|\, 2^i\geq g})$.
Since $h\geq g$, each box of $G_{\gamma/h}$ is occupied by at most one station -- its leader. We choose the leader
of each box of the pivotal grid by the algorithm GranLeaderElection (Algorithm~\ref{alg:gran}), which starts from assuming that all (active) stations are leaders
of respective boxes of $G_{\gamma/h}$ (note that there is at most one station in each box of this grid). Then,
it repeatedly applies the technique from Proposition~\ref{prop:lead2} in order to gradually obtain
leaders of larger boxes.


\begin{algorithm}[H]
	\caption{GranLeaderElection($V,g$)}
	\label{alg:gran}
	\begin{algorithmic}[1]
    \State $h\gets\min_{i\in\NAT}({2^i\,|\, 2^i\geq g})$
	\State  $x\gets r/h$;
    \State Each station $v\in V$ gets status leader of the appropriate box of $G_x$.
    \For{$i=1,2,\ldots,\log h$}
        \State   Choose leaders of boxes of $G_{2x}$ from
  leaders of $G_x$, using Proposition~\ref{prop:lead2}.
        \State $x\leftarrow 2\cdot x$
    \EndFor
    \end{algorithmic}
\end{algorithm}
Finally, we summarize properties of Algorithm GranLeaderElection in the following
proposition.
\begin{proposition}\labell{prop:leader}
Algorithm GranLeaderElection chooses the leader in each box of the pivotal
grid containing at least one element of $V$ in time
$O(\log g d_{\alpha}^2(n))$, where $d_{\alpha}$ is a flat function, provided granularity
of $V$ is not larger than $g$.
\end{proposition}

\subsubsection{Broadcasting Algorithm}
Given the algorithm electing the leaders in boxes of the pivotal grid,
we describe implementation of procedure {\Transmit}, called here Gran-{\Transmit}. In this way we obtain
algorithms {\AlgG}, which repeats Gran-{\Transmit} several times.

We say that a station $v$ is {\em $(d_1,d_2)$-connected}, for 
$(d_1,d_2)\in\DIR$ iff
$v\in C(i,j)$ for a box $C(i,j)$ of the pivotal grid and $v$
has a neighbor in the box $C(i+d_1,j+d_2)$ of the pivotal grid. Below, we formally
describe {\Transmit} procedure, which applies the leader election procedure
in order to transmit a message from each box containing stations in state
$\stwo$ to its neighbors. More precisely, for each direction $(d_1,d_2)\in\DIR$, 
the application
of leader election chooses one station $v$ in $C$ which has a neighbor in the box $C'$ located
in the direction $(d_1,d_2)$ from $C$ (if there is such a station in $C$)
and that station transmits successfully.
\tj{Then, the neighbor $u\in C'$ of $v$ with the smallest ID is chosen to broadcast
the message to all stations from $C$.}
In order to formalize this idea, assume that
$u,v$ are such stations that $u\in C'$ for a box $C'$ of the pivotal grid
and $u$ is in the range area of $v$. We say that $u$ {\em dominates box} $C'$ with respect to $v$
if $u=\min\{w\,|\, w\in C'\mbox{ and } w\mbox{ is in the range area of }v\}$.
\begin{algorithm}[H]
 \caption{Gran-\Transmit$(g)$}
 \label{alg:grantrans}
 \begin{algorithmic}[1]
  \For{$(d_1,d_2)\in\DIR$}
    \State $V_{(d_1,d_2)}\gets\{v\,|\, s(v)=\stwo\mbox{ and }v\mbox{ is } (d_1,d_2)\mbox{-connected}\}$
    \State GranLeaderElection$(V_{(d_1,d_2)},g)$ \Comment{\tj{leader of $(d_1,d_2)$-connected stations}}
	\State $d\gets d_{\alpha}(n)$, \Comment{$d_{\alpha}$ is a flat function from Corollary~\ref{cor:dilsuc}}
	  \For{$(j,k)\in[0,d-1]^2$}
	    \State \textbf{Round $1$}: a station $v$ transmits if it is elected
	    the leader
	    of its box (of the pivotal grid)
	    \State in step 3
	    during GranLeaderElection$(V_{(d_1,d_2)},g)$ and $v\in C(j',k')$
	    such that
	    \State $(j'\mod d,k'\mod d)=(j,k)$.
	    \State \textbf{Round $2$:} station $u$ transmits if: $s(u)=\sone$,
	    $u$ could hear $v$ in Round~$1$, $u\in C(j',k')$
	    \State such that
	    $((j'-d_1)\mod d,(k'-d_2)\mod d)=(j,k)$,
	    \State and $u$ dominates its box wrt $v$.
	  \EndFor
  \EndFor
  \State For each $v\in V$ such that $s(v)=\stwo$: $s(v)\gets \sthree$.
 \end{algorithmic}
\end{algorithm}


\begin{proposition}\labell{prop:gtrans}
Algorithm {Gran-\Transmit} works in time $O(d_{\alpha}^2(n)\log g)$ for a flat function
$d_{\alpha}:\NAT\to\NAT$ and it preserves properties (I) and (P).
\end{proposition}
\begin{proof}
Time complexity bound follows directly from Proposition~\ref{prop:leader} and
Corollary~\ref{cor:dilsuc}.

In order to prove (I), it is sufficient to show that in each box $C$ of the pivotal
grid and each execution of {Gran-\Transmit}, either all stations in $C$ move from the state
{\sone} to {\stwo}, or none station in $C$ changes its state from {\sone} to {\stwo} during
that execution of {Gran-\Transmit}.
Here we benefit from the fact that stations know their neighborhood. If a station $u$
from a box $C$ and in state $\sone$ receives a message from a station $v$ in state {\sthree},
and $u$ knows that $v$ transmits successfully,
then $u$ is also able to determine which other stations
in box $C$ receive the same message in the current round
(since it knows positions of these stations and $v$ sends its position inside a message).
In this way, the unique station $u$ (with smallest ID) among stations from box $C$ that have received the message from $v$ can be determined,
and this station transmits a message.
This
message is successfully heard by
all other stations in $C$ in the appropriate Round~2 (see line 9 of the algorithm),
since the set of stations sending messages
in Round~2 is $d$-diluted.
Assuming that all stations located in $C$ are in the state {\sone} at
the beginning of {\Transmit}, they change their states to {\stwo} at the end
of this execution of {\Transmit}.

As for (P), we make use of the fact that (I) is satisfied at the beginning
of each {\Transmit}. Thus, either all stations in a box $C$ are in state
{\stwo} at the beginning of {\Transmit} or none is.
In the former case, the
correctness of GranLeaderElection (see Proposition~\ref{prop:leader})
guarantees that if $C'$ is a neighbor of $C$ in direction $(d_1,d_2)$,
then a unique station $v$ from $C$ is chosen in line $3$,
which has a neighbor in $C'$
and then $v$
transmits
successfully in line $6$ (i.e., in Round~1, see Corollary~\ref{cor:dilsuc} for justification).
\end{proof}

Finally, we obtain the following result.
\begin{theorem}
Algorithm {\AlgG} performs broadcasting in a $n$-node network of diameter
$D$ with granularity $g$ in time $O(Dd_{\alpha}^2(n)\log g)$, where $d_{\alpha}$ is
a flat function.
\end{theorem}

\subsection{General Algorithm}
\labell{s:diam-gen}

In order to deal with networks with unlimited granularity, we propose a method
of ``decreasing'' granularity to the level of $2^{O(\log^2n)}$ in
time $O(\log^2 n)$. When granularity is decreased, we apply protocols
designed for networks with bounded granularity.

Our method of decreasing granularity applies a technique of simulating collision
detection in radio networks without collision detection, called Echo, c.f.,~\cite{KP04}.
Using a modified Echo procedure, we can choose ``representatives'' of dense areas
of a (box of a) network, which will work ``on behalf'' of whole such areas.
In this way we decrease granularity of the network. Importantly, this procedure
does not harm connectivity of the network nor changes its eccentricity more than by
a constant multiplicative factor. We describe this technique in Section~\ref{sub:echo}.

The above mentioned method of choosing representatives (of ``dense'' areas) works correctly when applied
to one set of stations such that each of them is in the range area of each other.
However, when one tries to apply it simultaneously to several remote groups of stations,
interferences incurred in the SINR model can disrupt these executions.
Therefore,
before applying the above method of decreasing granularity,
we first design an offline procedure --- based on the local views of stations ---
that partitions the set of stations
in a box of the pivotal grid into $\log n$ families of sets. (Note that each station
knows all elements of its box of the pivotal grid, since these stations are in its
range area.) The key property of this partition is that the sets in one family
$F$ (called {\em color}) are located
in such a way that one can execute the leader election procedure (i.e.,
the choice of representatives) based on Echo simultaneously on all sets from $F$.
Since each set in each family covers a square with side's length at least
$r/2^{O(\log^2 n)}$,
the leaders (representatives) elected in separated sets form
subnetworks with granularity $2^{O(\log^2 n)}$. This
local pre-processing
procedure is described in Section~\ref{sub:dilution}.

Finally, in Section~\ref{sub:genalg}, we provide algorithm {\AlgD}.
This algorithm follows the generic scheme described in Section~\ref{s:generic},
with additional local pre-processing (c.f., Section~\ref{sub:dilution})
and with specific implementations of {\Election} and {\Transmit}
based on the method of decreasing granularity described in Section~\ref{sub:echo}.

\subsubsection{Partition into collision avoiding families}
\labell{sub:dilution}

In the following, a {\em square} in the grid $G_a$ is a square
whose vertices belong to $G_a$ (thus the length of the side of each such
square is a multiplicity of $a$).
We associate such squares with stations of a network located in them in the following way:
\begin{itemize}
 \item[(a)]\labell{item:a}
 a square (box) $R$ of size $a\times a$ is associated with all stations located in it;
 \item[(b)]\labell{item:b}
 any larger square $R$ contains some subset of stations of the network located inside $R$; however,
for each square $R'$ of size $a\times a$ included in $R$, either $R$ contains all stations
of $R'$ or none of them.
\end{itemize}
Let $\m{S}$ be a set of squares in a grid $G_a$, each $R\in\m{S}$ has associated
a set of stations $V_R$ located inside $R$.
We say that $\m{S}$ is {\em collision
avoiding} if for each $R\in\m{S}$ and each $v\in V_R$, the following condition
is satisfied:
\begin{quote}
if the set of transmitting stations in a round is equal to
$\{v\}\cup\bigcup_{R'\in\m{S}\setminus \{R\}}V_{R'}$ \\
then the message of $v$
is received by each station from
$V_R\setminus\{v\}$.
\end{quote}
In other words, transmissions in squares different from $R$ cannot disrupt
communication in $R$ (even if all elements of other squares are transmitting
simultaneously), provided exactly one station from $R$ is transmitting.

Assume that
there are given
an upper bound $d\cdot a$ on the length of the side of a square
and
an upper bound $y$ on the number of stations associated with a square.
As we show in the following proposition, in order a set $\m{S}$ of squares
satisfying these bounds be collisions avoiding, it is sufficient that the
box-distance
between each two elements of $\m{S}$ is at least $d_{\alpha}(n)dy$, where
$d_{\alpha}$ is a flat function.


\begin{proposition}
\labell{prop:avoid}
For each $\alpha\geq 2$, there exists a 
flat function $d_{\alpha}$
satisfying the following property.
Let $\m{S}$ be a set of squares in a grid $G_a$, where $a=\gamma/c$ for some $c\in\NAT$,
such that
\begin{itemize}
\item
each square $R\in\m{S}$ has associated at most $y$ stations located inside $R$,
\item
the length of the side of each $R\in\m{S}$ is at most $d\cdot a$,
\item
for each $R_1,R_2\in\m{S}$, the box-distance between
$R_1$ and $R_2$ is not smaller than $x\cdot a$,
\item
the number of stations associated to all squares is equal to $n$,
\end{itemize}
for some $y,d,x\in\NAT_+$ such that $c>2d$.
If $x\geq d_{\alpha}(n)dy$ then
$\m{S}$ is collision avoiding.
\end{proposition}

\begin{proof}
Let $c\in\NAT$, $d,x,y\in\NAT_+$ be such that $c>2d$ and $x\geq d$
\tj{(note that the proposition concerns $x\geq d_{\alpha}(n)dy$ only)}.
Recall that $r=(1+\eps)^{-1/\alpha}$, $\gamma=r/\sqrt{2}$, $a=\gamma/c$.

Let $R\in\m{S}$ and $v\in R$. Since the side of $R$ is at most
$da=d\gamma/c$, 
the distance from $v$ to any other station $w\in R$
is at most $\sqrt{2}\gamma/c=dr/c$. Therefore the power of signal from $v$ received
by $w$ is at least
$$\frac1{(dr/c)^{\alpha}}=\left(\frac{c}{rd}\right)^{\alpha}.$$
On the other hand, $I$,
the total noise plu interference received by $w$ and caused by all elements of
$\bigcup_{R'\in\m{S}}V_{R'} \setminus V_R$
is at most
$$\cN+\sum_{j=1}^{n}4\cdot 5j\cdot\frac{y}{(j\cdot xa)^{\alpha}}\leq
1+d'_{\alpha}(n)\cdot\frac{y}{x^{\alpha}}\cdot\left(\frac{c}{r}\right)^{\alpha},$$
where $c'_{\alpha}=\max(1,20\cdot 2^{\alpha/2}\cdot \zeta(\alpha-1))$, $\zeta$
is the Riemann zeta function  
and $\cN=1$.
The above formula follows from the fact that there are at most
$20j$ squares such that the box-distance of each of them to
$R$ is in the interval $[j\cdot xa,(j+1)\cdot xa)$ and the
box-distance between each two of them is not smaller than $x\cdot a$
(see Figure~\ref{fig:election}).
\begin{figure}
\begin{center}
\epsfig{file=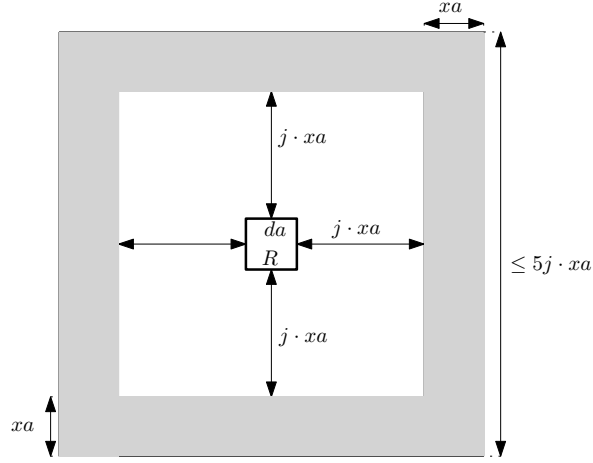, scale=0.8}
\end{center}
\caption{Illustration to the proof of Proposition~\ref{prop:avoid}. Each square whose
distance to $R$ is in $[j\cdot xa,(j+1)\cdot xa)$ has a nonempty intersection with
the gray frame. Moreover, the box-distance between any two such squares is at
least $xa$, the ``width'' of the frame.}
\label{fig:election}
\end{figure}%
Therefore,
$$\left(\frac{c}{rd}\right)^{\alpha}\geq 1+c'_{\alpha}\frac{y}{x^{\alpha}}\cdot\left(\frac{c}{r}\right)^{\alpha},$$
then the message from $v$ is received by $w$ if
This implies that the constraint
\begin{equation}\label{eq:xalpha}
x^{\alpha}\geq\frac{d'_{\alpha}(n)\cdot y}{\frac{1}{d^{\alpha}}-\left(\frac{r}{c}\right)^{\alpha}}
\
\end{equation}
gurantees that $w$ receives a message from $v$.
By the assumption
$c>2d$ and $r<1$, we see that
$$
{\frac{1}{d^{\alpha}}-\left(\frac{r}{c}\right)^{\alpha}}
>
\frac{2^{\alpha}-r^{\alpha}}{2^{\alpha}d^{\alpha}}
>
\frac{1}{2^{\alpha}d^{\alpha}}
\ ,
$$
and therefore
$$
\frac{d'_{\alpha}(n)\cdot y}{\frac{1}{d^{\alpha}}-\left(\frac{r}{c}\right)^{\alpha}}
<
d'_{\alpha}(n)\cdot 2^{\alpha}\cdot d^{\alpha}\cdot y
\ .
$$
Thus, if $x\geq d_{\alpha}(n)dy$ where $d_{\alpha}(n)=2\cdot (d'_{\alpha}(n))^{1/\alpha}$, then
the condition~(\ref{eq:xalpha}) for collision avoidance
is satisfied. 
%
\end{proof}

Below, we present algorithm {\NOGRAN} which splits a set of stations in $O(\log n)$
collision avoiding families of squares. More precisely, for each box $C$ of the pivotal
grid, the algorithm builds $\log n$ collision avoiding families of squares in $C$,
such that each station from $C$ belongs to some square in those families.

Let $C$ be a box of the pivotal grid.
We start with the set of squares of size $a\times a$
of the grid $G_a$ included in $C$ and containing at least one station, for some sufficiently small $a$
(line 2).
The goal is to build such a set of squares in each box of the pivotal grid
that subset of squares with similar ---
up to the multiplicative factor $2$ --- number of associated stations
is collision avoiding. In stages $i\in[0,\log n]$, we consider squares
with the number of associated stations in the interval $(2^{i-1},2^i]$
(see line 6)
and we keep an upper bound $d_i a$ on the length of the side of (so far unconsidered)
squares.
In each stage, we choose greedily as large as possible subsets of squares such
that each two squares of a subset are in large distance (to avoid interferences),
%
see lines 8-10
(c.f., Proposition~\ref{prop:avoid}).
These squares form the $i$th family of squares (color $i$).
The remaining squares are combined into larger
squares containing more than $2^i$ elements each
(see lines 7, and 11-13).
As we show, it is possible
to ensure that the upper bound on the lengths of the side of a square increases sufficiently
slow to guarantee that eventually each station belongs to some square and
the set of squares is split into $\log n$ collision avoiding families,
assuming $a=\gamma/2^{O(\log^2n)}$ (or $c=O(\log^n)$).

The key issue is that our ultimate goal is to guarantee that the set of squares with a fixed color
in all boxes (not only in one fixed box) are collision
avoiding, since the algorithm has to perform further computation in various
boxes simultaneously.
(By the way, if we restrict to one box of the pivotal grid, it is
sufficient to associate the same color to all stations. 
%
On the other hand, {\NOGRAN} is executed locally (in one box) since stations
should be able to perform this procedure without communication, on the basis of their
knowledge about neighborhood. One cannot exclude
that squares $R_1,R_2$ with the same color which belong to two adjacent boxes of the pivotal grid
are very close to each other.
Therefore, we refine our coloring in order to avoid the situation
that two squares from adjacent boxes have the same color (line 9).


\begin{algorithm}[H]
	\caption{\NOGRAN($C(j,k),c$)}
	\label{alg:nogran}
	\begin{algorithmic}[1]
    \State $a\gets \gamma/c\ \ (=r/(\sqrt{2}c))$
    \State $\m{S}\gets$ all nonempty boxes of $G_a$ inside the box $C$ of the pivotal grid 
    \State for each $R\in \m{S}$: $V_R\gets $ all stations located in $R$;
    \State $d_{0}\gets 1$
    \For{$i=0,1,\ldots,\log n$}
	\Comment{Iteration of phases}
        \State $x_i\gets c_{\alpha}d_{i}2^i$
        \State $W_i\gets\{R\in\m{S}\,|\,\, 2^{i-1}<|V_R|\leq 2^i\}$
        \State $E_i\gets\{(R_1,R_2)\,|\, R_1,R_2\in W_i, \distM(R_1,R_2)\leq x_i\cdot a\}$
        \For{each separated vertex $R$ of the graph $G_i(W_i,E_i)$}
            \State $\colour(R)\gets (i,j\mod 2,k\mod 2)$
            \State delete $R$ from $W_i$
        \EndFor
        \For{each connected component $W'\subseteq W_i$}
            \State Form a smallest square $R'$ containing all elements of $W'$, and add $R'$ to $\m{S}$
            \State Remove all elements of $W'$ from $\m{S}$
        \EndFor
        \State $d_{i+1}\gets 4(x_i+d_{i})$
    \EndFor
	\end{algorithmic}
\end{algorithm}
Now, we formally analyze algorithm {\NOGRAN}.
Let {\em phase $i$} denote the execution of the body of the main loop, i.e., lines 5-14,
of the algorithm {\NOGRAN} for the corresponding $i$.
Let $\side(R)$, for a square $R$, denote the length of the side of $R$.
We will show that the following invariants are satisfied at the beginning of the phase $i$, for every $i\geq 0$:
\begin{description}
\item[(A1)]
Each square $R\in\m{S}$ has more than $2^{i-1}$ stations (i.e., $|V_R|> 2^{i-1}$);
\item[(A2)]
For each $R\in\m{S}$, the length of the side of $R$ is not larger than
$\frac{|V_R|d_{i}}{2^i}\cdot a$.
\end{description}
\begin{proposition}
The algorithm {\NOGRAN} satisfies the invariants (A1) and (A2) at the beginning of
each phase.
\end{proposition}
\begin{proof}
The proof goes by induction. One can easily verify that the invariants are satisfied
at the beginning of phase $0$.
Next, assuming that the invariants are satisfied at the beginning of phase $i$, we show
that they are satisfied at the beginning of phase $i+1$ as well.

As for the invariant (A1), observe that each square having at most $2^i$ elements
is removed from $\m{S}$ during phase $i$ (in line 10 or 13). Moreover, each new square added to $\m{S}$
during phase $i$ contains stations of at least two removed squares
(see line 12 and the fact that each separated vertex/square $R$ is deleted in line 10).
Since (A1) is satisfied
at the beginning of phase $i$, the number of station in such a new square is larger than
$2^{i-1}+2^{i-1}=2^i$.

Concerning (A2), observe that a square that is in $\m{S}$ at the beginning of phase $i$
and is {\em not} removed from $\m{S}$ during phase $i$ satisfies the condition
$$
\side(R)\leq\frac{|V_R|d_{i}}{2^i}
\leq
\frac{|V_R|d_{i+1}}{2^{i+1}}
$$
at the
beginning of phase $i+1$, because
$d_{i}<d_{i+1}/2$ (line 14.).
Now, consider a square $R'$ added to $\m{S}$
during phase $i$.
Let $W'$ be the connected component of $W_i$ whose elements form $R'$.
Let $x_1,x_2$ ($y_1,y_2$, respectively) be the smallest and largest values
of the first (second, respectively) coordinate of vertices of squares from $W_i$.
W.l.o.g. assume that $x_2-x_1\geq y_2-y_1$. Thus, $\side(R')=x_2-x_1$. Then,
there exists a path $(R_1,\ldots,R_p)$ in $W'$ such that $x_1$ is the first
coordinate of some vertex of $R_1$, $x_2$ is the first coordinate of some vertex
of $R_p$. Our inductive assumptions imply that:
\begin{itemize}
\item
$2^i\geq|V_{R_j}|>2^{i-1}$ for each $j\in[p]$;
\item
$\side(R_j)\leq\frac{|V_{R_j}|d_{i}}{2^i}\cdot a\leq d_{i}\cdot a$
for each $i\in[p]$;
\item
$\distM(R_j,R_{j+1})\leq x_j\cdot a$ for each $j\in[p-1]$;
\item
$\sum_{j=1}^p\side(R_j)+\sum_{j=1}^{p-1}\distM(R_j,R_{j+1})\geq x_2-x_1=\side(R')$.
\end{itemize}
Thus,
$$
\side(R')\leq \left(pd_{i}+(p-1)x_i\right)\cdot a
\leq
p(x_i+d_{i})a
=
pd_{i+1}a/4
$$
and
$$|V_{R'}|>p\cdot 2^{i-1}.$$
Therefore,
$\side(R')\leq pd_{i+1}a/4
<
\frac{|V_{R'}|}{2^{i-1}}\cdot\frac{d_{i+1}}4\cdot a
=
\frac{|V_{R'}|d_{i+1}}{2^{i+1}}\cdot a$,
which confirms that the invariant (A2) is satisfied at the beginning of phase $i+1$.
\end{proof}

\remove{

\begin{proposition}\labell{prop:simgram}
There exists a constant $c$ which depends only on $\alpha$ such that, if $c\geq 2^{c_1\log^2n}$
then the set of squares which have assigned color $(i,j,k)$ by {\NOGRAN} is collision avoiding for
each $(i,j,k)\in[\log n]\times[0,1]^2$.
\end{proposition}

}

\begin{proposition}
\labell{prop:simgram}
There exists a constant $c_1$, which depends only on $\alpha$, such that:
if $c\geq 2^{c_1\log^2n}$
then the set of stations with assigned color $i$ by \NOGRAN($\cdot,c$) is collision avoiding, for each $i\in[\log n]$.
\end{proposition}

\begin{proof}
First, assume that all stations are located in one box of the pivotal grid.
The choice of $x_i$ in Algorithm {\NOGRAN} (line 5) guarantees that the set of squares
with color $(i,j,k)$ is collision avoiding due to Proposition~\ref{prop:avoid}, provided
$c>2d_i$. Since
$d_{0}=1$, $d_{i+1}=4(x_i+d_i)$ and $x_i=c_{\alpha}d_i2^i$, the relationship
$$
d_{i+1}
=
4(d_{i}+x_i)
=
4d_{i}(1+c_{\alpha}2^i)
\leq
8c_{\alpha}2^i\cdot d_{i}
$$
holds for $i\geq 0$, where the last inequality follows from the fact that
$c_{\alpha}\geq 1$ (see Proposition~\ref{prop:avoid}).
Thus,
$$d_{i+1}
\leq
(8c_{\alpha})^{i+1}\prod_{j=0}^i2^j
=
2^{i(i+1)/2+(i+1)\log(8c_{\alpha})}
\ .
$$
Therefore $d_{\log n}=2^{O(\log^2 n)}$ and the appropriate choice of $c_1$ guarantees
that $d_{\log n}<2^{c_1\log^2n}/2$.
So, the proposition holds for $c=2^{c_1\log^2 n}$,
since $c>2d_i$ for each $i\in[\log n]$ and squares with each color are
collision avoiding by Proposition~\ref{prop:gtrans}.

Now, consider the case when stations are located in various boxes of the pivotal grid.
The choice of colors guarantees that $\distM(R_1,R_2)\geq r$ for any two squares
$R_1,R_2$ with color $(i,j,k)$ such that $R_1\in C$, $R_2\not\in C$, where $C$
is a box of the pivotal grid (the method of assigning $j,k$ guarantees that
$R_1$ and $R_2$ are not in adjacent boxes). In order to guarantee the correctness
of the proposition, it is sufficient that $\distM(R_1,R_2)\geq c_{\alpha}d_{\log n}2^{\log n}$
(see Proposition~\ref{prop:avoid}).
Since 
$\distM(R_1,R_2)\geq r=(1+\eps)^{-1/\alpha}$,
it is enough to assure that $c_{\alpha}d_{\log n}2^{\log n}<r$,
which can also be guaranteed for $c=2^{O(\log^2n)}$.
\end{proof}
\tj{
Finally, we can state the key property of the algorithm {\NOGRAN}.
\begin{lemma}\labell{l:nogran}
Algorithm {\NOGRAN} forms the set of $O(\log n)$ collision avoiding
families of squares such that each station belongs to (exactly) one
square in these families.
\end{lemma}
\begin{proof}
Since there are $n$ stations overall, $|V_R|\leq n$ and therefore each station
is assigned to a square. Proposition~\ref{prop:simgram} implies that
those families are collision avoiding. Finally, it follows directly from the
algorithm that each station belongs to exactly one square from those families.
\end{proof}
}

\subsubsection{Election by Echo}
\labell{sub:echo}

Before we specify exactly how our application of the procedure Echo \cite{KP04} works,
let us explain what is the task we would like to solve by using this technique. During {\Transmit},
if stations in a box $C$ of the pivotal grid are in state {\stwo}, the goal is to send a message
to at least one station
in each box $C'$ of $C$ which is a neighbor of $C$. To inform a station in $C'$,
it is sufficient that exactly
one station from $C$ that has a neighbor in $C'$ is transmitting in some step successfully.
We are going to assure this property by guaranteeing that exactly one station
is transmitting among stations having neighbors in $C'$.
However, although each station
from $C$ knows whether it has a neighbor in $C'$, it does not necessarily know which
other stations from $C$ have also neighbors in $C'$.
%

The goal of the algorithm {\AlgCRBE} is as follows. We are given a set $V_1$ of stations
such that $(v,w)$ is an edge in the communication graph, for each $v,w\in V_1$ and the set
$V_1$ is known to each $v\in V_1$. Moreover, $V_2\subseteq V_1$ is defined such that each
$v\in V_1$ knows whether
it belongs to $V_2$ (i.e., whether $v\in V_2$),
but 
it may not have a knowledge
which of the remiaining elements
of $V_1$ belong to $V_2$. As a result, a unique {\em representative} $w$ of $V_2$ should be chosen
and all elements of $V_1$ should be aware of $w$;
in case of
$V_2=\emptyset$, all elements of $V_1$
should be aware of that fact.
\begin{algorithm}[H]
	\caption{\AlgCRBE($V_1, V_2$)}
	\begin{algorithmic}[1]
    \State $\psi\gets\min_{v\in V_1}(v)$
    \State $\psi$ transmits a message with information whether $\psi\in V_2$
    \If{$\psi\in V_2$} return $\psi$ and finish \EndIf
    \State Let $\varphi\in V_1$ be $\min_{w\in V_1\setminus\{\psi\}}\{w\,|\,\dist(\psi,v)\leq\dist(\psi,w)\}$
    every $v\in V_1$
    \State $\varphi$ transmits a message with information whether $\varphi\in V_2$
    \If{$\varphi\in V_2$} return $\varphi$ and finish \EndIf
    \State Assign unique temporary IDs (TIDs) in $[|V_1|]$ to all elements of $V_1$:
		TID$(v)\gets |\{u\in V_1\,|\,u\leq v\}|$
    \State $bot\gets 1$; $top\gets |V_1|$
    \While{$bot\leq top$}
        \State $mid\gets \lfloor (bot+top)/2\rfloor$
        \State $T\gets\{v\in V_2\,|\, bot\leq\text{TID}(v)\leq mid\}$
        \State {\bf Round $R_1$:} each $v\in T$ transmits the message $m_v$ encoding $v$
        \State {\bf Round $R_2$:} each $v\in T\cup\{\varphi\}$ transmits the message $m_v$ encoding $v$
        \State {\bf Round $R_3$:}
            \State\hspace{\algorithmicindent} \textbf{if} $\psi$ can hear $m_v$ in $R_1$ for $v\in V_1$ \textbf{then}
		$\psi$ transmits $m_v$ 
            \State\hspace{\algorithmicindent} \textbf{else} \textbf{if} {$\psi$ can hear $m_{\varphi}$ in $R_2$}
			\textbf{ then } $\psi$ transmits $m_{\varphi}$
        \State\textbf{if} {$m_{v}$ is heard in $R_3$ for $v\in V_1\setminus\{\varphi\}$}\textbf{ then } 
		return $v$ and finish the algorithm's execution
	\State\textbf{if} {$m_{\varphi}$ is heard in $R_3$}\textbf{ then }
		$bot\gets mid+1$
        \State\textbf{else} {$top\gets mid$}
    \EndWhile
    \end{algorithmic}
\end{algorithm}

Just for further consideration we would like to point out that  $\dist(\psi,\varphi)$ is the
largest among distances between between elements of $V_1$. This implies that $\psi$ can hear
$\varphi$ only when no other $x\in V_1$ is transmitting a message.
\begin{proposition}\labell{prop:chooserepr}
Assume that the algorithm {\AlgCRBE} is executed in parallel on a family of collision avoiding squares. 
Then, each execution of \AlgCRBE($V_1,V_2$)
finishes in $O(\log n)$ rounds and it
gives the following result:
\begin{itemize}
\item
if $V_2=\emptyset$: each station of $V_1$ knows that $V_2$ is empty;
\item
otherwise, each $v\in V_1$ knows a fixed station $w\in V_2$ called
a {\em representative} of $V_2$.
\end{itemize}
\end{proposition}

\begin{proof}
As for time complexity, note that $top-bot$ becomes roughly twice smaller in each execution
of the loop 9-19 (see lines 18-19).

The assumption that the algorithm is executed on collision avoiding squares
implies that we can assume that each execution of \AlgCRBE($V_1,V_2$)
satisfies the following condition: if exactly one element of $V_1$ (different from
$\varphi$) transmits
a message in a round, then this message is received by all elements of $V_1$.
Moreover, since $\dist(\psi,\varphi)\geq\dist(\psi,v)$ for each $v\in V_1$,
$\psi$ cannot receive a message from $\varphi$ if any element of $V_1\setminus\{\varphi\}$ transmits
a message at the same round. These observations imply that, after round $R_3$,
all stations from $V_1$ can determine whether the subset of $V_2$, which consists
of stations with TIDs in the range $[bot,mid]$, contains $0$, $1$, or more than one
element. Thanks to this fact, an execution of lines 10-19 gives each element
of $V_1$ information whether  $X=V_2\cap\{v\,|\, TID(v)\in[bot,mid]\}$
is empty.
Using this property, the while-loop 9-19 applies
binary search in order to choose
a representative of $V_2$, if $V_2\neq\emptyset$.
More precisely, if $m_{\varphi}$ is heard in $R_3$ then $X$ is empty
(and searching is restricted to the range $[mid+1,top]$), and it is
not empty otherwise.
%
%
\end{proof}




\subsubsection{Broadcasting Algorithm}
\labell{sub:genalg}

Finally, we define a broadcasting algorithm {\AlgD}, which repeats several times
the algorithm {Gen-\Transmit} given below. Algorithm {Gen-\Transmit} resembles the algorithm
{Gran-\Transmit} from Section~\ref{s:granularity-unknown}.
However {Gen-\Transmit}, first applies the technique of ``decreasing'' granularity
introduced in Sections~\ref{sub:dilution} and~\ref{sub:echo}.


%

Recall the following definitions.
A station $v$ is {\em $(d_1,d_2)$-connected} for $d_1,d_2\in\{0,1,2\}$ iff
$v\in C(i,j)$ for a box $C(i,j)$ of the pivotal grid and $v$
has a neighbor in the box $C'(i+d_1,j+d_2)$ of the pivotal grid.
Let $u,v$ be such stations that $u\in C$, for a box $C$ of the pivotal grid,
and $u$ is in the range area of $v$. We say that $u$ dominates $C$ with respect to $v$
if $u=\min\{w\,|\, w\in C\mbox{ and } w\mbox{ is in the range area of }v\}$.
We also set $g_{\alpha}=2^{c_1\log^2 n}$, where $c_1$ is the constant
from Proposition~\ref{prop:simgram}.
\begin{algorithm}[H]
	\caption{Gen-\Transmit}
	\begin{algorithmic}[1]
    \State Each station $v$ in state {\stwo}  executes
    \NOGRAN($C, c_{\alpha}$), where $C$ is the box of the pivotal grid
    \tj{containing $v$, $c_{\alpha}\gets d_{\alpha}(n)$ and $d_{\alpha}$ is a
    flat function satisfying properties stated in Proposition~\ref{prop:avoid}} 
    \For{each $(d_1,d_2)\in\DIR$}
        \For{$(i,j,k)\in[\log n]\times\{0,1\}\times\{0,1\}$}
            \For{each square $R$ of color $(i,j,k)$ {\bf in-parallel}}
                \State \AlgCRBE($V_R, V_R\cap \{v\,|\, v\mbox{ is }(d_1,d_2)\mbox{-connected}\}$)
            \EndFor
        \EndFor
        \State GranLeaderElection($\{v\,|\, v\mbox{ is a representative chosen in line 5}\}, g_{\alpha}$)
	\State 
            $d\gets$ parameter from Corollary~\ref{cor:dilsuc} applied for the set of leaders of boxes of $G_{\gamma}$.
	  \For{$(j,k)\in[0,d-1]^2$}
	    \State \textbf{Round $1$}: A station $v$ transmits if:
	    \State $v$ is elected
	    the leader
	    of its box of the pivotal grid
	    in line 6
	    during GranLeaderElection,
	    \State and $v\in C(j',k')$
	    such that
	    $(j'\mod d,k'\mod d)=(j,k)$
	    \State \textbf{Round $2$:} A station $u$ transmits if:
	    \State $s(u)=\sone$,
	    \State $u$ heard $v$ in Round $1$, 
	    \State $u\in C(j',k')$ such that
	    $((j'-d_1)\mod d,(k'-d_2)\mod d)=(j,k)$,
	    \State and $u$ dominates its box wrt $v$
	  \EndFor
    \EndFor
\end{algorithmic}
\end{algorithm}

\begin{proposition}\labell{prop:gentrans}
Algorithm {Gen-\Transmit} works in time $O(d_{\alpha}^2(n)\log^2 n)$ for a flat function $d_{\alpha}$
and it preserves the properties
(I) and (P) from page~\pageref{i:P}.
\end{proposition}
\begin{proof}
As for time complexity, the execution of ChooseReprByEcho in line 5 requires $O(\log n)$ rounds,
and the execution of GranLeaderElection in line 6 requires $O(\log^2n)$ rounds. Since
$d$ and the size od $\DIR$ are constant, {Gen-\Transmit} works in time $O(\log^2 n)$.

As algorithm {Gen-\Transmit} follows the structure of {Gran-\Transmit},
the fact that it preserves
(I) and (P) can be proved similarly \tj{as Proposition~\ref{prop:gtrans}}.
In fact, it is sufficient to prove that if
there is $v\in C$ in state active for a box $C$ which is $(d_1,d_2)$-connected,
then $C$ has the leader after step 6. This claim is a consequence of the
following facts:
\begin{itemize}
\item
$\NOGRAN(C,g_{\alpha})$ in line 1 guarantees that each active station
which is $(d_1,d_2)$-connected is associated with some square which
has assigned a color in $[\log n]\times\{0,1\}^2$; moreover, squares
with the same color are collision avoiding (Proposition~\ref{prop:simgram});
\item
{\AlgCRBE} (line 5) chooses a representative of
$V_R\cap \{v\,|\, v\mbox{ is }(d_1,d_2)\mbox{-connected}\}$ for each square
$R$, provided  $V_R\cap \{v\,|\, v\mbox{ is }(d_1,d_2)\mbox{-connected}\}\neq\emptyset$
thanks to the fact that squares with a fixed color are collision avoiding
(Proposition~\ref{prop:chooserepr});
\item
Granularity of the set of representatives in line 6 of the algorithm is at most
$g_{\alpha}$ by Proposition~\ref{prop:chooserepr} and item (b) on page~\pageref{item:b}
defining restrictions on associations of squares with stations. Therefore,
GranLeaderElection in line 6 chooses the leader in the box $C$, if
the set of station from $\{v\,|\, v\mbox{ is a representative chosen in line 5}\}$
located in $C$ is nonempty (Proposition~\ref{prop:leader}).
\end{itemize}
%
\end{proof}

Below, we state a theorem which follows directly
from the specification of Algorithm {\AlgD} (i.e., repeating algorithm {Gen-\Transmit})
and from Proposition~\ref{prop:gentrans}.

\begin{theorem}\labell{t:brodcast:known:gen}
Algorithm {\AlgD} performs broadcasting in a $n$-node network of diameter
$D$ in time $O(Dd_{\alpha}^2(n)\log^2 n)$, where $d_{\alpha}$ is a flat function.
\end{theorem}

\section{Size Dependent Algorithm for Anonymous Networks}
\labell{s:anonymous}

In this section we consider fully anonymous ad hoc networks in which, at the
beginning of a protocol, execution each station knows only $n$, $N$, its
own ID and its position in the Euclidean space (i.e., its
coordinates). We develop a deterministic broadcasting algorithm
{\AlgN}, which matches the lower bound $\Omega(n\log N)$ (see Theorem~\ref{t:lower:log}).

\subsection{High-Level Idea of Algorithm {\AlgN}}

Our algorithm executes repeatedly two threads.

The first thread keeps
combining stations into groups in such a way that eventually, for any box $C$ of
the pivotal grid, all stations located in $C$ form one group.
Moreover, each group should have the leader, and each station should be aware of
(i) which group it belongs to, (ii) which station is the leader of that group, and
(iii) which stations belong to that group (i.e., a station should know the set of IDs
and positions\footnote{\tj{It is sufficient that $O(\log^2n)$ bits of coordinates
of stations are stored.}} of all stations in the group).
These properties are achieved as follows.

Upon waking up, each station forms a group with a single element (itself),
and then the groups increase gradually by merging.
The merging process builds upon the following observation.
Let $\sigma$ be the smallest distance between two
stations taking part in the first thread, and let $u,v$ be two closest
stations. Thus, there is at most one transmitting
station in each box of the grid $G_{\sigma/\sqrt{2}}$.
Then, if $u$ ($v$, resp.) transmits a message
and no other station in distance $d\cdot \sigma$, for some constant $d$,
transmits at the same time,
then $v$ ($u$, resp.) can hear that message (see Proposition~\ref{prop:avoid}).
Using combinatorial structure called {\em strongly-selective family} (ssf) as
a broadcast schedule, one can assure that a round satisfying
these properties occurs in $O(\log n)$ rounds. If $u$ can hear $v$
and $v$ can hear $u$ during such a schedule, the groups of $u$ and $v$
can be merged into one larger group.

The second thread, on the other hand, is supposed to guarantee that in each round $t$
of the algorithm and for each group of stations $H$, exactly one station
from $H$ is transmitting a message in round $t$. This property will be satisfied
provided each station knows its group, so it can determine its temporary
ID (TID) as the rank of its ID in the sequence of IDs of stations from
the group, taken in a nondecreasing order. Using these TIDs, the stations of the group
apply round-robin strategy. Thus, if each group
corresponds to all stations in the appropriate box, transmissions in
the second thread are successful (see Corollary~\ref{cor:dilsuc}, Proposition~\ref{prop:avoid}
for $y=1$, $a=1$ and $d=1$),
and therefore they guarantee that all neighbors of the box will have
informed stations, provided 
the second thread is executed for sufficiently long time.

In order to apply the above described ideas for global broadcasting, it is
necessary to repeat Threads~1 and 2 several times.
The main problem with implementation
and its analysis is
that there is no simple way to determine whether group(s) already
covers the whole box of the pivotal grid. Moreover, as long as there are many groups inside
a box, transmissions in the second thread may cause unwanted interferences.
Another problem is that the set of stations attending the protocol changes gradually,
when new stations become informed and can
join
the execution of the protocol.
Therefore we modify the above described ideas in the following way:
\begin{itemize}
\item
The two threads --- one forming groups and the other transmitting in a round-robin fashion ---
are interleaved such that one round of the former is followed by one
round of the latter. This will be conceptually implemented in a form of two parallel threads.
\item
In order to tackle the lack of knowledge about the progress in computation,
each station participates in the protocol for $T(n)$ rounds, where $T(n)$
is the upper bound on the round complexity of accomplishing our broadcasting
algorithm derived in the analysis.
\item
Finally, our proof of complexity bound is based on measuring the progress of computation at round $t$
by using amortized analysis, in a way reflecting the advancement of the process
of merging groups and receiving the broadcast message by consecutive
stations.
%
\end{itemize}

\subsection{Formal Implementation of Algorithm {\AlgN}}

Each station $v$ keeps in its local memory a boolean variable $L(v)$ indicating
whether $v$ has the status of the {\em leader} of its group, and local variables
$M(v)\in V$ and $G(v)\subseteq V$. Let us think of a directed graph defined by edges
$(v,M(v))$.
Our goal is to preserve the invariant that
the graph is a forest $F$ and each edge $(v,M(v))$ is directed \tj{from a child to its parent
in the appropriate tree of $F$}.
Provided this invariant is preserved, we define $master(v)$ as the
transitive closure of $M(v)$, i.e., $master(v)=v$ if $M(v)=v$ and $master(v)=master(M(v))$ otherwise.
Moreover, $group(v)=G(master(v))$.
The fact that pointers $M(v)$ define a forest gives a partition of the set of stations in the following way:
\begin{itemize}
\item
each tree of this graph forms one group;
\item
each group has the leader which is equal to the root of the appropriate tree; that is, the leader
of the group to which $v$ belongs is equal to $master(v)$.
\end{itemize}

We say that a station $v$ is {\em consistent} if $M(v)=master(v)$ and $G(v)=group(v)$.
Initial values of the local variables of stations are as follows:
$L(v)\gets true$, $M(v)\gets v$, $G(v)\gets \{v\}$. Thus, all
stations are consistent at the beginning.
%
A {\em leader} is each station $v$ such that $L(v)=$true.

We say that a network satisfies {\em integrity at time $t$} iff
\begin{enumerate}
\item[(a)]
groups $G(v)$ known by leaders 
at the end of round $t$ form a partition of the set of all stations $V$ (i.e., $V=\bigcup_{\{v\,|\,L(v)\}}G(v)$
and $G(v)\cap G(u)=\emptyset$ for each $v\neq u$ such that $L(v)=L(u)=true$);
\item[(b)]
$G(v)\subseteq G(M(v))$ for each station $v$;
\item[(c)]
$M(v)\in\boxx(v)$ and $G(v)$ contains only stations located in $\boxx(v)$.
\end{enumerate}
One of invariants which we are going
to \tj{be preserved along executions} of {\AlgN} is that all leaders are consistent, and the network
satisfies integrity.
%
Ideally, we would also like to achieve consistency of stations which are not leaders ---
unfortunately this property will not be guaranteed by our solution,
however our algorithm will be able to achieve it at some crucial stages of the broadcasting task.
%

The algorithm proceeds in two parallel threads:
Thread~1 and Thread 2. We assume
that Thread~1 is executed in odd rounds (i.e., in rounds $t$ such that $t\mod 2=0$) and Thread
2 in even rounds. In order to simplify presentation, we assume that rounds of Thread~1/Thread~2
have consecutive numbers $1,2,3,\ldots$ Below, we describe both threads in more detail.

\paragraph{Thread 1.}
The main goal of Thread~1 is to merge groups such that \tj{consistency of leaders
and integrity of network} are preserved. The following technical proposition is the key for guaranteeing
process of merging groups is fast enough.

\begin{proposition}
\labell{prop:for:sel}
For each $\alpha>2$, there exists a constant $d$, which depends
only on the parameters $\eps,\beta$ and $\alpha$ of the model,
satisfying the following property.
Let $W$ be a set of stations such that there
is at most one station from $W$ in each box
of the grid $G_x$, for some $x\leq \gamma$, and
$\min_{u,v\in W}\{\dist(u,v)\}=x\cdot\sqrt{2}\}$. 
If station $u\in C$ for a box $C$ of $G_x$ is transmitting in a round $t$ and no other station
in any box $C'$ of $G_x$ in the box-distance at most $d$ from $C$ is transmitting at that round,
then $v$ can hear the message from $u$ at round $t$.
\end{proposition}
\begin{proof}
Let $u,v$ satisfy properties stated in the proposition.
If $u$ is transmitting in round $t$ then the power of the signal of $u$ arriving
at $v$ is
\begin{equation}\label{e:tr}
\frac{1}{(\sqrt{2}x)^{\alpha}}
\geq
(1+\eps)\cN
\ ,
\end{equation}
where the inequality follows from the fact that $\sqrt{2}x\leq r=(1+\eps)^{-1/\alpha}$
(recall that we assume $\beta=1$).
Observe that, under the assumptions of the proposition,
the number of stations whose distance to $v$ is in the interval $[ix,(i+1)x)$ is
not larger than the number of boxes of $G_x$ in box-distance $i$ from the box
containing $v$, which in turn is equal to $8(i+1)$.
Assuming that no station
in any box $C'$ in the box-distance at most $d$ from $C$ is transmitting,
the amount of interference and noise at $v$ is smaller than
$$
\cN+\sum_{i=d}^{\infty}8(i+1)\cdot\frac{1}{(ix)^{\alpha}}
=
\cN+\frac{8}{x^{\alpha}}\cdot c_d
\ ,
$$
where $c_d=\sum_{i=d+1}^{\infty}i^{1-\alpha}$. 
Thus, by (\ref{e:tr}) it is sufficient to show that
there exists $d$ which guarantees that
$$\cN+\frac{8}{x^{\alpha}}c_d\leq(1+\eps)\cN\mbox{ or } \cN+\frac{8}{x^{\alpha}}c_d\leq\frac{1}{2^{\alpha/2}x^{\alpha}}$$
for each $x>0$, which is equivalent to:
\begin{equation}\label{e:cases}
c_d\leq \frac{1-\cN(\sqrt{2}x)^{\alpha}}{8\cdot 2^{\alpha/2}}\mbox{ or } c_d\leq \frac{\eps\cN x^{\alpha}}{8}.
\end{equation}
Consider two cases:

\noindent Case A: $\cN(\sqrt{2}x)^{\alpha}\leq \frac12$

This case reduces the first inequality of (\ref{e:cases}) to $c_d\leq \frac{1}{16\cdot 2^{\alpha/2}}$ which
is satisfied for sufficiently large $d$, due to convergence of $\sum_{i} i^{1-\alpha}$.

\noindent Case B: $\cN(\sqrt{2}x)^{\alpha}> \frac12$

In this case, the second inequality of (\ref{e:cases}) reduces to
$c_d\leq \frac{\eps}{16\cdot 2^{\alpha/2}}$ which
is also satisfied for sufficiently large $d$, due to convergence of $\sum_{i} i^{1-\alpha}$.

\comment{ 
$$
\cN+\frac{8}{x^{\alpha}}\cdot c_d
<
(1+\eps)\cN
\ ,\mbox{ or }
$$
then $v$ receives the message from $u$. Since $x\leq \sqrt{2}r<\sqrt{2}$, the above
inequality is satisfied if $c_d<\frac{\eps\cN2^{\alpha/2}}{8}$, which in turn
is satisfied for sufficiently large $d$ depending on $\alpha$
because $\sum_{i=1}^{\infty}i^{1-\alpha}=\zeta(\alpha-1)$
is bounded by a constant, for $\alpha>2$.
} 
\end{proof}


A family $S=(S_0,\ldots,S_{s-1})$ of subsets of $[N]$ is a {\em $(N,k)$-ssf (strongly-selective family)} of length $s$
if, for every non empty subset $Z$ of $[N]$ such that
$|Z|\leq k$ and for every element $z\in Z$, there is a set $S_i$ in $S$ such that
$S_i\cap Z=\{z\}$.
It is known that there exists $(N,k)$-ssf of size $O(k^2\log N)$ for every $k\leq N$,
c.f.,~\cite{ClementiMS01}.
Let $k=(2d+1)^2$, let $S$ be a $(N,k)$-ssf,
and let $s=|S|=O(\log N)$.
%
The sets $S_0,\ldots,S_{s-1}$ of the family $S$ define a broadcast schedule in such a way
that station $v$ transmits in round $t$ iff $v\in S_{t\mod s}$ (formally, the bit $t$
of $S(v)$ is equal to 1 iff $v\in S_t$).
\begin{corollary}\labell{cor:selector}
For each $\alpha>2$, there exists a constant $d$, which depends
only on the parameters $\eps,\beta$ and $\alpha$ of the model,
satisfying the following property.
Let $W$ be a set of stations such that
$\min_{u,v\in W, \boxx(u)=\boxx(v)}\{\dist(u,v)\}=x$ and let $\dist(u,v)=x$
for some $u,v\in W$ such that $\boxx(u)=\boxx(v)$ and $W$ is \tj{$d$-diluted for
$d\geq 2$.}
Then, $v$ can hear the message from $u$ during an execution of a $(N,k)$-ssf on $W$.
\end{corollary}

\tj{Now, we are ready to describe Thread 1 in detail. Given a $(N,k)$-ssf $S$ of length $s$,}
Thread 1 consists of blocks of $2s$ rounds, each block split in two
stages of length $s$. Importantly, a station which becomes informed
during a block, starts participating in the execution of the protocol
in the next block of Thread 1. Algorithm~\ref{alg:un:n:ini} describes behavior
of a station $v$ in step $t$. Note that the initial value of $X_v$ is equal to the
empty set for each $v$ at the beginning of a block \tj{(see Algorithm~\ref{alg:un:n:mod})}
and then, it is equal to the set of station
which transmitted successfully to $v$ during the block.
\begin{algorithm}[H]
	\caption{Thread1($v,t$)}
	\label{alg:un:n:ini}
	\begin{algorithmic}[1]
    \State $t'\gets t\mod 2s$
\If{$v$ informed before step $t-t'$}\Comment{$v$ informed before the current block}
    \If{$t'<t$ } \Comment{(Stage 1 of a block)}
    \If{$L(v)$ and $v\in S_{t\mod s}$}
        \State $v$ transmits a message including $v$ and $G(v)$
    \Else
        \If{$L(v)$}
            \If{$v$ can hear $u$} $X_v\gets X_v\cup \{u\}$
	    \EndIf
        \Else
            \If{$v$ can hear $u$ such that $G(v)\subset G(u)$} $M(v)\gets u$; $G(v)\gets G(u)$
	    \EndIf
        \EndIf
    \EndIf
    \Else \Comment{(Stage 2 of a block)}
    \If{$L(v)$ and $v\in S_{t\mod s}$}
        \State $v$ transmits a message including $v$ and $X_v$
    \EndIf
    \EndIf
\State Modify($v,t$)
\EndIf
    \end{algorithmic}
\end{algorithm}
In a single block of Thread~1, the $(N,k)$-ssf $S$ is executed twice: once in Stage~1
and once in Stage~2. At the end of the block, the procedure Modify is executed, whose
goal is to merge groups using information gathered in Stages~1 and 2 of the current
block.
In Stage~1, each station $v$ determines $X_v$, the set of stations $u$ such that $v$ can hear $u$ during
the execution of $S$ (on the set of stations active at the beginning of Stage~1 of the block).
In Stage~2, each station $v$ sends $X_v$, and in this way, at the end of Stage~2, it also collects information about $X_u$ for each $u\in X_v$.


For a fixed block of computation, let $G'(V,E')$ be a symmetric graph
which consists of such edges $(u,v)$ that \tj{$u$ and $v$ have the status of leaders,}
$u$ can hear $v$ and
$v$ can hear $u$ during the block of computation. Note that $(u,v)\in E'$ iff
$v\in X_u$ and $u\in X_v$. Thus, each station can determine its
neighbors in $G'$ at the end of each block (since $v$ knows $X_v$
after Stage~1,
and it learns $X_u$, for each $u\in X_v$, during Stage~2).

At the end of each block of Thread~1, each station modifies
its local variables appropriately, by executing procedure Modify, c.f., the pseudo-code
of Algorithm~\ref{alg:un:n:ini}.
The goal is to make
at least one merge of two groups.  In order to achieve this goal,
we implement an algorithm which builds
(in distributed way) a matching in $G'$ such that the matching
is nonempty iff the set of edges of $G'$ is nonempty as well.
\tj{(Actually, our algorithm builds such a matching that each station $v$ satisfying the following
properties chooses its ``partner'' in the matching: $v$ can
hear another station during a block and $v$ is smaller than IDs of stations which transmitted
successfully a message to $v$ in the block.)}
Then,
the groups of the pairs of stations in the matching are merged.
\begin{algorithm}[H]
	\caption{Modify($v,t$)}
	\label{alg:un:n:mod}
	\begin{algorithmic}[1]
\If{$t\mod 2s=0$} \Comment{Execute at the end of round $t$ such that $t\mod 2s=0$}
    \State $match(v)\gets nil$
    \If{$L(v)$ and $X_v\neq\emptyset$}
        \State $u\gets\min(X_v)$
        \If{$v=\min(X_u)$}
            \State $match(v)\gets u$
            \If{$v>u$}
                \State $M(v)\gets u$; $L(v)\gets false$
            \EndIf
            \State $G(v)\gets G(v)\cup G(u)$
        \EndIf
    \EndIf
    \State $X_v\gets\emptyset$
\EndIf
    \end{algorithmic}
\end{algorithm}

\paragraph{Thread 2.}
In Thread~2, each station applies round-robin algorithm inside its group.
This is done successfully provided the stations possess up to date information about their groups ---
which is the goal of the previously described Thread 1.
\begin{algorithm}[H]
	\caption{Thread2($v,t$)}
	\label{alg:un:n:th2}
	\begin{algorithmic}[1]
    \State $\Delta\gets |G(v)|$
    \State $TID(v)\gets |\{u\,|\,u\in G(v)\mbox{ and }u< v\}|$
    \State if $t\mod\Delta=TID(v)$: $v$ transmits a message.
    \end{algorithmic}
\end{algorithm}

\subsection{Analysis}
Recall that we make a simplifying assumption that, if at most one station
from each box of the pivotal grid transmits in a round $t$, then
each such transmission is successful.
Due to Corollary~\ref{cor:dilsuc}, one can achieve this property
using dilution with constant parameter $d$ (provided $\alpha>2$),
which does not change the asymptotic complexity
of our algorithm.

First, we prove some basic properties of Thread~1.

\begin{proposition}
\labell{prop:cons}
Thread~1 preserves consistency of leaders and integrity of network at any round.
\end{proposition}

\begin{proof}
Assume that consistency of leaders and integrity of network are satisfied
at the beginning of a block of Thread 1.
Since variables determining integrity of the network and consistency
of stations change only at the end of blocks (i.e., during the execution of algorithm
Modify), let us
consider round $t$ at the end of a block.
Note that $u=match(v)$ iff
$v=match(u)$ at the end of $Modify(v,t)$. Moreover, if $u=match(v)$ and
$v=match(u)$, then exactly one of $u,v$ becomes non-leader and one of
them remains the leader. Thus, as a result, the groups $G(v), G(u)$
are replaced by $G(v)\cup G(u)$ after step $t$, which proves integrity. Since the
group of the station $v$ changes only in case $u=match(v)$,
$v=match(u)$ and $L(v)=L(u)=true$ for some $u$,
it preserves consistency thanks to the fact that
such $u$ and $v$ exchange messages with $u$ during the analyzed block of
Thread~1.
\end{proof}

We say that {\em station $u$ joins the group of station $v$} during the block of Thread 1
if $L(u)=L(v)=true$ at the beginning of the block, while
$L(u)=false$, $L(v)=true$, and $M(u)=v$ at the end of that block.

\begin{lemma}
\labell{l:merge}
Assume that the set $W$ of leaders at the beginning of a block of
Thread~1 contains at least two elements, 
which are located in the same box of the pivotal grid. Then, there
exist $u,v\in W$ such that $u$ joins the group of $v$ during the
block.
\end{lemma}

\begin{proof}
\tj{Let $y$ be equal to the smallest distance between a pair of stations
$u,v\in W$ such that $u$ and $v$ belong to the same box of the pivotal grid.
Let $u,v$ be the elements of $W$ such that $\dist(u,v)=y$ and $\boxx(u)=\boxx(v)$.
Let $x=y/\sqrt{2}$.}
Let $u\in C$ for a box $C$ of the grid $G_x$ and let
$A$ be the set of elements of $W$ located in boxes of $G_x$ which are in
box-distance at most $d$ from $C$, where $d$ is the constant from
Proposition~\ref{prop:for:sel}.
The set $A$ contains at most $(2d+1)^2$ elements, since each box of $G_x$ contains
at most one element of $W$.
Therefore, there
exists a round
$t\le s$
in the ssf $S$ such that $v$ is transmitting
a message at round $t$ and no other element of $A$ is transmitting at that round.
Proposition~\ref{prop:for:sel}
implies that $u$ can hear $v$ in such a round. Similarly, $v$ can
hear $u$ during an execution of $S$. Therefore, there exists at least
one pair $(u,v)$ such that $u\in X_v$ and $v\in X_u$ at round
$2s$
of the block, which is equivalent to the fact that
%
$E'=\{(u,v)\,|\,u\in X_v \mbox{ and } v\in X_u\}$, the set
of edges of a graph $G'(V,E')$, is nonempty.
Now, let $u$ be the smallest ID of a
node
whose degree in $G'$ is larger than zero.
Let $v$ be its neighbor in $G'$ with the smallest ID.
It is clear from the construction that $v$ joins the group of $u$ in such case
(see algorithm Modify$(v,t^*)$, for $t^*$ being the last round of the block).
\end{proof}

In general, it might happen that a station which is not a leader
is not consistent.
Such a situation occurs, for example, when $u$ joins the group of $v$
and then $v$ joins the group of $w$. Simultaneously, while $v$ can hear $w$
when it joins the group of $w$, it is possible that $u$ cannot
hear $w$.
The following lemma states that eventually, when there
is at most one leader in each box at the beginning of a block of Thread~1,
then for each leader, all stations in its box correctly update the information
about their masters and groups
during the considered block and become consistent.

\begin{lemma}
\labell{l:masters}
Assume that there is 
\tj{at most} one leader in each box of the pivotal grid containing active stations,
at the beginning of a block of Thread~1. Then,
for each box $C$ containing a leader and each $v\in C$ that is informed at the beginning
of the block, $v$ is consistent at the end of the block.
\end{lemma}

\begin{proof}
Let $v\in C$ be informed and let $u\in C$ be
the only leader in $C$ at the beginning of a block.
Integrity of the network and consistency of leaders
\tj{(Proposition~\ref{prop:cons})}
guarantee that $u=master(v)$.
The station $v$ can hear $u$ during the block,
which follows from the fact that each leader broadcasts successfully during the
block (due to our simplifying assumption concerning situation that at most one
station in each box of the pivotal grid is transmitting). 
Thus, since $v$ receives a message from $u=master(v)$, it updates its
local variables in line 10 of 
pseudo-code of Thread~1 and becomes consistent.
\end{proof}
We say that a block $j$ of Thread~2 is {\em partially stable} if the following
conditions are satisfied:
\begin{itemize}
\item
each box of the pivotal grid contains at most one leader;
\item
at least one informed station is not consistent;
\end{itemize}
at the beginning of the block $j$.
Formally, we define progress of algorithm {\AlgN} at 
the end of block $j$
as $\pi(j)$, equal
to the sum of the following four components:
\begin{itemize}
\item[(a)]
the number of informed stations;
\item[(b)]
$n$ minus the number of groups;
\item[(c)]
the number of tuples $(v,d_1,d_2)$
such that $v$ is an informed station, $d_1,d_2\in\DIR$,
$v$ belongs to $C(i,l)$ for some $i,l\in\INT$,
and there is an informed station in the box $C'=C(i+d_1,l+d_2)$,
where $C,C'$ are boxes of the pivotal grid;
\item[(d)]
\tj{the number of partially stable blocks
of Thread~1 up to round $t$.}
\end{itemize}
It is clear that the \tj{expressions} described in the above items (a)--(c) have
always values in $O(n)$. We show that (d) is also in $O(n)$, which directly implies
that $\pi(j)=O(n)$ for every $j$.

\begin{proposition}
For each network with $n$ stations, the number of \tj{partially stable} blocks of
Thread 1
\comment{ 
satisfying the following conditions
\begin{enumerate}
\item[(i)]
each box contains at most one leader at the beginning of the block;
\item[(ii)]
at least one informed station is not consistent at the beginning of the block;
\end{enumerate}
} 
is smaller than $n$.
\end{proposition}

\begin{proof}
Consider two consecutive blocks $j_1<j_2$ of
Thread 1
satisfying (i) and (ii).
Lemma~\ref{l:masters} implies that all stations informed at the beginning of
block $j_1$ are consistent at the end of this block. \tj{Note that} an informed station
located in a box $C$ of the pivotal grid with one leader may
loose its
consistency only in the case
when a new station from box $C$ becomes informed. Since there is an informed station
that is not consistent at the beginning of block $j_2$ (c.f., (ii)), the number
of informed stations at the beginning of block $j_2$ is larger than the number
of informed stations at the beginning of block $j_1$. Therefore the number
of blocks of
Thread~1
satisfying (i) and (ii) is smaller than $n$.
\end{proof}

Now, we show that
the amortized increase of cost $\pi$ during each {\em block of Thread~1} ---
defined as the time period including block of Thread~1 and rounds of Thread~2 interleaved with
the block of Thread~1 --- is at least one.


In the following, we analyze progress of computation during blocks of Thread~1, however
we take into account also rounds of Thread~2 occurring during the time span of the analyzed
block of Thread~1 (recall that the executions of the two threads are interleaved).

\begin{lemma}
\labell{l:progress}
Assume that some stations are not yet informed at the beginning of some
block $j$ of Thread~1. Then, there exists a block $k\geq j$ 
such that the total increase
of progress function in blocks $j,j+1,\ldots,k$ is at least $k-j+1$.
\end{lemma}

\begin{proof}
If there are two informed stations $u,v\in C$, for a box $C$ of the pivotal grid,
such that $L(u)=L(v)=true$
(i.e., $u,v$ are leaders)
at the beginning of block $j$, 
progress increase is guaranteed in block $j$
by Lemma~\ref{l:merge},
since at least one merge of two groups takes place.

If there is at most one leader in each box at the beginning of block $j$,
then we consider two cases:
\begin{description}
\item[Case 1.]
All informed stations are consistent at the beginning of block $j$.\\
In this case all transmissions in both Threads are successful, as long as the
number of informed stations does not change.
Therefore, each informed station can transmit successfully. And, since not all stations
in the network
are informed and the network is connected, a new station becomes informed eventually.
%
Let $k\geq j$ be the
smallest number of a block 
in which a new station $v$ becomes informed.

If this station $v$ belongs
to a box which has an informed leader at the beginning of block $j$, then
$v$ becomes informed in block $j$ and the progress increase is $1$ in round $j$, which
certifies the claimed result for $k=j$. 

If the box $C'$ containing station $v$ does not have an informed leader at the beginning
of block $j$, then $C'$ does not have any informed station at the beginning of block $j$
either (due to integrity of the network).
Let $u\in C$ be a station that informed $v$ and $k\geq j$ be
the number of the block in which $v$ becomes informed. Since each transmission
of Thread~2 is successful in this case, and Thread 2 applies a round-robin protocol
on stations from $C$,
$u$ does not transmit in blocks $j,j+1,\ldots,k-1$ implies that
the number of stations in box $C$ is at least $(k-j)+1$ (since at 
least one station from box $C$ transmits during the time span of one block in Thread~2).
\tj{Moreover, there are $k-j$ various stations in $C$ such that each of them
transmits successfully in blocks $j,\ldots,k-1$. Let $C=C(i,j)$, $C'=C(i+d_1,j+d_2)$.
}
Therefore, the
number of tuples
$(v,d_1,d_2)$
such that $v$ is an informed station and belongs to $C(i,j)$ for some $i,j\in\INT$,
$d_1,d_2\in\DIR$
and there is an informed station in the box $C'(i+d_1,j+d_2)$,
increases by at least
$k-j+1$
throughout blocks $j,\ldots,k$. Therefore, the progress $\pi$ increases by at least $k-j+1$.

\item[Case 2.]
There is a station which is {\em not} consistent at the beginning of block $j$.\\
Then, the part (d) of the potential function $\pi$ increases until the end
of block $j$, according to Lemma~\ref{l:masters}.

%
%
\end{description}
\end{proof}

Finally, we obtain the following theorem as a direct consequence of Lemma~\ref{l:progress}.
\begin{theorem}
Algorithm {\AlgN} performs broadcasting in each $n$-node network in time $O(n\log N)$.
\end{theorem}

\section{Degree Dependent Algorithm for Anonymous Networks}\labell{s:algdeg}

In this section we present a broadcasting algorithm which achieves complexity $O(D\Delta\log^2N)$
in anonymous networks, i.e., when neighborhood is not known.

The core of the algorithm is a leader election procedure which, given a set of stations $V$,
chooses exactly one station (the leader) in each box $C$ of the pivotal grid which contains
at least one element of $V$. This procedure works in $O(\log n\cdot \log N)$ rounds \tj{and it is
executed several times}.
The set of stations attending a particular leader election execution consists of all stations
which received the broadcast message and have not bo chosen leaders of their boxes in previous
executions of the leader election procedure.
Moreover, at the end of each execution of the leader election procedure, each leader chosen
in that execution transmits a message successfully (see Corollary~\ref{cor:dilsuc}). In this way, each station receives
the broadcast message after $O(D\Delta\log^2 N)$ rounds.

\subsection{Leader Election}
In the following, we describe the leader election algorithm. We are given a set of stations $V$
of size at most $n$. The set $V$ is not known to stations, each station knows merely whether it
belongs to $V$ or it does not belong to $V$. In the algorithm, we use  $(N,d)$-ssf $S$ of size
$s=O(\log N)$, where $d$ is the constant from Proposition~\ref{prop:for:sel}. As before, $X_v$ for a given
execution of $S$ is defined as the set of stations which belong to $\boxx(v)$ and
$v$ can hear them during that execution. The key observation for our construction is in fact
a consequence of Corollary~\ref{cor:selector}.

\begin{proposition}\labell{prop:closer}
For each $\alpha>2$, there exists a constant $k$, which depends
only on the parameters $\eps,\beta$ and $\alpha$ of the model,
satisfying the following property.
Let $W$ be a $3$-diluted (wrt the pivotal grid) set of stations and let
$C$ be a box of the pivotal grid. If
$\min_{u,v\in C\cap W}
%
%
=x\leq 1/n$ and $\dist(u,v)=x$
for some $u,v\in W$ such that $\boxx(u)=\boxx(v)=C$,
then $v$ can hear the message from $u$ during an execution of a $(N,k)$-ssf on $W$.
\comment{
For each $\alpha>2$, there exists a constant $k$, which depends
only on the parameters $\eps,\beta$ and $\alpha$ of the model,
satisfying the following property.
Let $W$ be a $3$-diluted (wrt the pivotal grid) set of stations and let
$C$ be a box of the pivotal grid. If
$\min_{u,v\in W, \boxx(u)=\boxx(v)=C}\{\dist(u,v)\}=x\geq 1/n$ and $\dist(u,v)=x$
for some $u,v\in W$ such that $\boxx(u)=\boxx(v)=C$,
then $v$ can hear the message from $u$ during an execution of a $(N,k)$-ssf on $W$.
}
\end{proposition}
\begin{proof}
Let $u,v$ and $x$ be as specified in the proposition and let $C=\boxx(u)=\boxx(v)$.
Let $S$ be a $(N,k)$-ssf.
If all stations
from $W$ are located in $C$, then the claim follows directly from Corollary~\ref{cor:selector}.
So, let $W'$ be the set of all elements of $W$ which are {\em not} located in $C$.
Let us (conceptually) ``move'' all stations from $W'$ to boxes adjacent to $C$,
preserving the invariant that $\min_{u,v\in W, \boxx(u)=\boxx(v)=C}\{\dist(u,v)\}=x$. Note that such
a movement is possible, since there are at most $n$ stations in $W'$ and the side of a box of the
pivotal grid is larger
than $1/2$. Since $W$ is $3$-diluted, the distance from $w\in C$ to any station $w'\in W'$ before
movement of $w'$ is larger than the distance from $w$ to $w'$ after movement.
Let $W''$ define $W$ with new locations of stations (after movements).
Therefore, if $u$ can
hear $v$ in the execution of $S$ on $W''$ (i.e., after movements of stations), it can hear $v$
in the execution of $S$ on $W$ (i.e., with original placements of stations).
However, the fact that $u$ can hear $v$ on $W''$ follows directly from the fact that
$\min_{u,v\in W''}\{\dist(u,v)\}=x$ by Corollary~\ref{cor:selector}.
\end{proof}

\comment{
The leader selection algorithm consists of two stages. The first stage gradually eliminates
the set of candidates for the leader in consecutive executions of a selector $S$ in the first
for loop. Therefore, we call this stage {\em Elimination}.
Let {\em block} $l$ of Elimination stage denote the executions of $S$ for $i=l$.
Each ``eliminated'' station $v$ has
assigned the value $ph(v)$ which is equal to the number of the block in which it is eliminated.
Let $V(l)=\{v\,|\, ph(v)>l\}$ and $V_C(l)=\{v\,|\, ph(v)>l\mbox{ and } \boxx(v)=C\}$ for $l\in\NAT$
and $C$ which is a box of the pivotal grid. The key property of sets $V_C(l)$
is that $|V_C(l+1)|\leq |V_C(l)|/2$ and the granularity of $V_C(l_C^{\star})$ is smaller than $n$
for each box $C$ and $l\in\NAT$, where $l_C^{\star}$ is the largest $l\in\NAT$ such that $V_C(l)$
is not empty (see Lemma~\ref{l:lead:empty} below).
Thus, in particular, $V_C(l)=\emptyset$ for each $l\geq \log n$.
Therefore, we can try to choose the leader of
each box $C$ applying (simultaneously in each box) the granularity dependent leader election algorithm
on $V_C(\log n)$, $V_C(\log n-1)$, $V_C(\log n-2)$ and so on, until the leader of $C$ is chosen. This
idea is implemented in the second part of the algorithm, called {\em Selection}. Now, we provide the
pseudo-code of the leader election algorithm and then its correctness and complexity are formally
analyzed.
}

The leader election algorithm consists of two stages. The first stage gradually eliminates
elements from the set of candidates for the leader in consecutive executions of a selector $S$ in the first
for loop. Therefore, we call this stage {\em Elimination}.
Let {\em block} $l$ of Elimination stage denote the executions of $S$ for $i=l$.
Each station $v$ ``eliminated'' in block $l$ has
assigned the value $ph(v)=l$.
Let $V(l)=\{v\,|\, ph(v)>l\}$ and $V_C(l)=\{v\,|\, ph(v)>l\mbox{ and } \boxx(v)=C\}$ for $l\in\NAT$
and $C$ which is a box of the pivotal grid. The key property of sets $V_C(l)$
is that $|V_C(l+1)|\leq |V_C(l)|/2$ and the granularity of $V_C(l_C^{\star})$ is smaller than $n$
for each box $C$ and $l\in\NAT$, where $l_C^{\star}$ is the largest $l\in\NAT$ such that $V_C(l)$
is not empty. 
Therefore, we can choose the leader of
each box $C$ applying (simultaneously in each box) the granularity dependent leader election algorithm
on $V_C(l_C^{\star})$.
It is done by the second stage, which applies the granularity dependent leader election
on $V_C(\log n)$, $V_C(\log n-1)$, $V_C(\log n-2)$ and so on, until the leader of $C$ is chosen.
After it is done all stations in $C$ become silent.
This
idea is implemented in the second part of the algorithm, called {\em Selection}.
Now, we provide the
pseudo-code of the leader election algorithm and then its correctness and complexity are formally
analyzed.
\begin{algorithm}[H]
	\caption{LeaderElection($V,n$)}
	\label{alg:leader}
	\begin{algorithmic}[1]
    \State For each $v\in V$: $cand(v)\gets true$; 
    \For{$i=1,\ldots,\log n+1$}\Comment{Elimination}
        \For{$j,k\in[0,2]$}
            \State Execute $S$ twice on the set:
            \State $\{w\in V\,|\,cand(w)=true \mbox{ and }w\in C(j',k')$ \mbox{ such that } $(j'\mod 2,k'\mod 2)=(j,k)\}$;
            \State Each $w\in V$ determines and stores $X_w$ during the first execution of $S$ and
            \State $X_v$ for
                each $v\in X_w$ during the second execution of $S$,
            \For{each $v\in V$}
                \State $u\gets \min(X_v)$
                \If{$X_v=\emptyset$ or $v>\min(X_u\cup\{u\})$}
                    \State $cand(v)\gets false$; $ph(v)\gets i$
                \EndIf
            \EndFor
        \EndFor
    \EndFor

    \State For each $v\in V$: $state(v)\gets active$ \Comment{Selection}
    \For{$i=\log n,(\log n)-1,\ldots,2,1$}
        \State $V_i\gets$ GranLeaderElection($\{v\in V\,|\, ph(v)=i, state(v)=active\},1/n$)\Comment{$V_i$ -- leaders}
        \State Each element $v\in V_i$ sets $state(v)\gets leader$ and transmits successfully 
        \State \tj{using constant dilution (see Corollary~\ref{cor:dilsuc})}
        \State Simultaneously, for each $v\in V$ which can hear $u\in\boxx(v)$: $state(v)\gets passive$
    \EndFor
    \end{algorithmic}
\end{algorithm}

\begin{lemma}\labell{l:lead:empty}
Let $C$ be a box of the pivotal grid and $l\in\NAT$.
Then,
\begin{enumerate}
\item
$|V_C(l+1)|\leq |V_C(l)|/2$;
\item
If $V_C(l+1)$ is empty, then the smallest distance between elements of $V_C(l)$ is at least $1/n$.
\end{enumerate}
\end{lemma}
\begin{proof}
Similarly as in Section~\ref{s:anonymous}, our algorithm implicitly builds matchings
in the graphs whose vertices are $V_C(l)$ and an edge connects such $u$ and $v$ that
$u$ can hear $v$ and $v$ can hear $u$ during an execution of $S$.
Note that the station $v\in V_C(l)$ belongs to $V_C(l+1)$ only if
the following conditions are satisfied:
\tj{
\begin{itemize}
\item
$v=\min(X_u)$;
\item
$u=\min(X_v)$;
$v<u$
\end{itemize}
for some $u\in V_C(l)$. That is, only elements of the matching belong to
$V_C(l+1)$ and exactly one element from each matched pair belongs to $V_C(l+1)$.}

Therefore, the inequality $|V_C(l+1)|\leq |V_C(l)|$ holds. This gives
item 1 of the lemma.

\tj{As for item 2, assume that $V_C(l)$ is not empty. Observe that $V_C(l+1)$ is
not empty if there exist $v,u\in V_C(l)$ such that $v$ can hear $u$
and $u$ can hear $v$. (Indeed, $v\in V_C(l+1)$ for the smallest $v\in V_C(l)$ such that $v$ can hear $u$
and $u$ can hear $v$ for some $u\in V_C(l)$.) However, such $v$ and $u$ exist if the smallest
distance between elements of $V_C(l)$ is at least $\frac1{n}$ by Proposition~\ref{prop:closer}.
}
\end{proof}

\begin{theorem}\labell{t:leader:general}
Algorithm LeaderElection chooses the leader in each box of the pivotal grid containing
at least one element of $V$ in $O(\log n\log N)=O(\log^2 N)$ rounds, provided $\alpha>2$.
\end{theorem}
\begin{proof}
Time complexity $O(\log^2n)$ follows immediately from the bounds on the size of selectors
and complexity of GranLeaderElection.

Lemma~\ref{l:lead:empty}.1 implies that $V_C(l)=\emptyset$ for each box $C$ and $l>\log n$.
(In other words, $ph(v)\leq\log n$ for each $v\in V$.)
Moreover, by Lemma~\ref{l:lead:empty}.2, the smallest distance between stations of $V_C(l_0)$ is at least $1/n$,
where $l_0=\max_l\{V_C(l)\neq\emptyset\}$. In other words the smallest distance between stations of
$\{v\in V\,|\, ph(v)=l_0, state(v)=active\}$ is $\geq 1/n$, where $l_0$ is the largest
number $l$ such that $ph(v)=l$ for some $v\in V$.

Let us focus on a box $C$ which contains at least one station from $V$.
Selection stage (the for-loop in lines 13-16) tries to choose the leader of $C$ among
$V_C(\log n), V_C(\log n-1), \ldots$. Moreover, when the leader is elected, all stations
from $C$ are switched off (i.e., their state is set to passive which implies that they
do not attend further GranLeaderElection executions).
Since $l_0=\max_l(V_C(l)\neq\emptyset)\leq\log n$
and the smallest distance between elements of $V_C(l_0)$ is $\geq 1/n$, each execution of GranLeaderElection
is applied on a set of stations with the smallest distance between stations $\geq 1/n$, and therefore the leader
in each box $C$ containing (at least one) element of $V$ is chosen by LeaderElection.
\end{proof}

\subsection{Broadcasting Protocol}
Algorithm~\ref{alg:broadcast:general} implements our broadcasting algorithm which repeats leader
election procedure several times and each station is ``switched off'' after it is elected a leader of its box
(assuring that each leader $v$ transmits the broadcast message successfully to all station
accessible from $v$).
\begin{algorithm}[H]
	\caption{GeneralBroadcast($V,n$)}
	\label{alg:broadcast:general}
	\begin{algorithmic}[1]
    \State The source transmits the broadcast message
    \State $V_1\gets \{v\in V\,|\, v\mbox{ received the broadcast message}\}$
    \For{$i=1,2,\ldots,D\Delta$}
        \State LeaderElection($V_{i},n$)
        \State $V_{i+1}\gets\{v\in V\,|\, state(v)\neq leader, v\mbox{ received the broadcast message}\}$
    \EndFor
    \end{algorithmic}
\end{algorithm}
\begin{theorem}\labell{t:broadcast:general}
Algorithm GeneralBroadcast finishes broadcasting in $O(D\Delta\log^2 N)$ rounds in ad hoc networks, provided that
$\alpha>2$ and each station knows $N,D$ and $\Delta$.
\end{theorem}
\begin{proof}
Let $P$ be a shortest path in the network graph from the source to a station $v$. Then, the length
of $P$ is at most $D$. Theorem~\ref{t:leader:general} guarantees that each station $v$
is elected a leader of its box $C$ after at most $\Delta$ executions of LeaderElection
following the execution in which $v$ receives the broadcast message. Moreover, a station elected
the leader of its box successfully sends the broadcast message to all its neighbors in the
network graph. Therefore, the broadcast message arrives to the last vertex of $P$ in
$O(D\Delta\log^2N)$ rounds.
\end{proof}
In order to implement Algorithm GeneralBroadcast, the knowledge of $n$, $D$ and $\Delta$ is required.
However, if $n$ is not known, one can implement LeaderElection in $O(\log^2N)$ rounds using the
bound $n\leq N$. Moreover, each station $v$ which is elected a leader of its box in GeneralBroadcast,
does not attend the protocol after the execution of LeaderElection
in which it is chosen a leader.
And, each station is eventually elected a leader.
Therefore, instead of the for-loop repeated $D\Delta$ times, it is sufficient that each station
participates in the protocol until its state changes to the value $leader$. This observation leads
to the following corollary.
\begin{corollary}\labell{c:broadcast:general}
One can build a protcol which finishes broadcasting in $O(D\Delta\log^2 N)$ rounds in ad hoc networks, provided that
$\alpha>2$ and each station knows merely $N$.
\end{corollary}

\section{Lower Bounds} \labell{s:lower}
\tj{ 
In this section we provide lower bounds which are close
to the the upper bounds provided so far. (In fact, they
leave the gap $O(\log N)$ in most cases.)
}

For a network  with distinguished source station $s$, $L_i$ denotes the set
of nodes in distance $i$ from $s$ in the communication graph (thus, in particular, $L_0=\{s\}$
and $L_1$ is equal to the set of neighbors of $s$).

\begin{theorem}\labell{t:lower:log}
There exists an infinite family of networks requiring
$\Omega(n\log N)$ rounds in order to accomplish deterministic
ad hoc broadcasting in the SINR model without local knowledge.
\end{theorem}
\begin{proof}
First, we describe a family of networks $\m{F}$ such that
broadcasting in SINR requires  time $\Omega(D\log N)$.

Each element of $\m{F}$ is formed as a sequential composition
of $D$ networks $V_1,\ldots,V_D$ of eccentricity $3$ each, such that:
\begin{itemize}
\item
the source $s$ is connected with two nodes $v_1,v_2$ in $L_1$ with arbitrary
IDs;
\item
$v_1,v_2$ are connected with $w$, the only element of $L_2$, and satisfy
the condition:
\begin{equation}\label{e:dist}
P\cdot\dist(v_1,w)^{-\alpha}=P\cdot\dist(v_2,w)^{-\alpha}-\cN/2.
\end{equation}
\end{itemize}
Moreover, we assume that  $\beta=1$. Finally, sequential composition of networks
$V_1,\ldots,V_D$ stands for identifying \tj{the element $w$} of network component $V_i$
with the source $s$ of network component $V_{i+1}$.

Note that if $v_1$ and $v_2$ transmit simultaneously in a network component $V_i$, the
message is {\em not} received
by $w$. Using simple counting argument, one can force
such choice of IDs of $v_1$ and $v_2$ that $\Omega(\log N)$ rounds are
necessary until a round in which exactly one of $v_1,v_2$ transmits
a message under the SINR model.
Since $D=\Theta(n)$ in the above construction,
the bound $\Omega(n\log N)$ holds.
\end{proof}



\remove{
For the purpose of lower bounds, we assume no background noise;
this allows us to simplify the presentation by using ranges instead of the whole SINR expressions.
Note however that the same results can be proved
using the analogous arguments (with additional constant scaling to overcome the impact of fixed
background noise) for any fixed value of background noise.
}

\begin{theorem}\labell{th:lower}
For any deterministic broadcasting algorithm $\cA$
and for every $D\geq 3$ and $\Delta\ge 4$, there exists a network of at most $D\Delta$ nodes with eccentricity
$D$ and maximal degree $\Delta$ on which algorithm $\cA$ completes broadcasting in $\Omega(D\Delta)$
rounds.
\end{theorem}

\begin{proof}
Let $\gammapom=1/\sqrt{2}$.
Let $\m{F}$ be a family of networks $F_j$, for $1\le j\le\Delta$,
of eccentricity $3$ which consist of three layers:
\begin{itemize}
\vspace*{-1ex}
\item
the source $s$,
located in the origin point $(0,0)$,
is the only element of $L_1$;
\vspace*{-1ex}
\item
$L_2$ consists of $\Delta$ nodes $v_0,\ldots,v_{\Delta-1}$, where the position of $v_i$
is $(\gammapom\cdot \frac{i}{\Delta}, \gammapom)$ for $0\le i\le\Delta-1$;
\vspace*{-1ex}
\item
$L_3$ contains only one node $w_j$ with coordinates
$(\gammapom\cdot\frac{j}{\Delta},\gamma+1)$.
\end{itemize}
Thus, the family $\m{F}$ consists of $\Delta$ elements, each network $F_j\in \m{F}$ is
uniquely determined by the
value $j$ fixing the position of node $w_j\in L_3$.

In what follows, we assume that the ranges of $s$ and $v_0,\ldots,v_{\Delta-1}$ are
equal to $1$.
Then,
\begin{itemize}
\vspace*{-1ex}
\item
$v_0,\ldots,v_{\Delta-1}$ are in the range area of $s$;
\vspace*{-1ex}
\item
$w_j$ is in the range of $v_j$ and it is not in the range of any other station from $L_1\cup L_2$;
\vspace*{-1ex}
\item
if more than $2^{\alpha/2}$ stations from $L_2$ transmit in a round, node $w_j$ cannot hear a message.
\end{itemize}
\vspace*{-1ex}
The first two bullets follow directly from the location of points and the value of range.
The last bullet holds because the minimum (maximum) of the distances between $v_i$ and $w_j$ is larger
than or equal to $1$ (smaller than $\sqrt{2}$), which guarantees that 
$SIR(v_i,w_j,\m{T})$ 
is smaller than $1$ for each of the transmitting stations $v_i$
if $|\m{T}\cap L_2|\geq 3$, where $\m{T}$ is the set of transmitting stations.

Consider any broadcasting algorithm $\cA$. 
We specify an adversary
who simultaneously, round after round, decides
what is heard by stations in $L_1\cup L_2$ in consecutive rounds of $\cA$
and restricts the family of considered networks $\m{F}$ to the networks on which
such answers are valid.
The goal of the adversary is to
prevent the arrival of a message to $w_j$ as long as possible.
Assume that the source sends the broadcast message to all nodes in $L_1$ in round $0$.
The adversary determines the family $\m{F}_t$, for every $t\leq \lfloor\Delta/2\rfloor-1$, in the following
way:
\begin{enumerate}
\vspace*{-1ex}
\item $\m{F}_0\leftarrow \m{F}$
\vspace*{-1ex}
\item $c\gets \lceil 2^{\alpha/2}\rceil$
\vspace*{-1ex}
\item For $t=1,2,\ldots, \lfloor \Delta/c\rfloor -1$ do:
 \begin{enumerate}
\vspace*{-1ex}
 \item
 if $v_{i_1}, \ldots, v_{i_{c'}}$ are the only stations from $L_1$ that transmit
 a message in the $t$-th round of $\cA$
 on the networks from $\m{F}_{t-1}$, and $c'\leq c$ \\
 then
 $\m{F}_t\leftarrow \m{F}_{t-1}\setminus\{F_{i_1},\ldots,F_{i_c}\}$;
\vspace*{-.5ex}
 \item
 otherwise, $\m{F}_t\leftarrow \m{F}_{t-1}$.
 \end{enumerate}
\end{enumerate}
\vspace*{-1ex}
One can easily verify that, for each $t\leq \lfloor \Delta/c\rfloor -1$,
the following conditions are satisfied:
\begin{itemize}
\vspace*{-1ex}
\item
$\m{F}_t$ is not empty;
\vspace*{-1ex}
\item
the history of communication
(i.e., messages/noise heard by all stations in consecutive rounds)
is the same in each network from $\m{F}_t$ up to the round $t$;
\vspace*{-1ex}
\item
$w_j$ does not
receive the broadcast message by round $t$ in the execution of $\cA$ on any network in~$\m{F}_t$.
\end{itemize}
\vspace*{-1ex}
This provides the claimed lower bound for constant eccentricity $D$.
In order to generalize this bound for arbitrary $D$, one can
consider a family of networks which consists of $(D-1)/2$ networks
from $\m{F}$
shifted
such that the source of the $i$th
network is equal to the only element in layer
$L_3$
in the $(i-1)$st
network,
for $2\le i\le (D-1)/2$.
The above strategy of the adversary
can be applied sequentially to every subsequently shifted network from $\m{F}$, to gain
the multiplicative factor $D$.
Note also that the size of the obtained network is $\frac{D-1}{2}\cdot (\Delta+1)+1\le D\Delta$,
its maximum degree is $\Delta$ and its eccentricity is $\frac{D-1}{2}\cdot 2 +1=D$.
\remove{
Now, consider the nonuniform known model~(b). First, let us concentrate
on the networks from $\m{F}$. If the range of
\dk{$w_j\in L_3$}
is smaller than
$1$ then the position of $w_j$ is unknown to the nodes of $L_1\cup L_2$
\dk{in network $F_j$, for any $1\le j\le \Delta$,}
and the above lower bound $\Omega(\Delta)$ applies for the family $\m{F}$.
Note that the largest distance between $s$ and any element
\dk{in $L_2$}
is $d=d(s, v_{\Delta-1})<1$. Therefore, if the range of $w_j$ in network $F_j$
from $\m{F}$ is equal to $d$, we can combine $D$ networks from $\m{F}$ as
above for model~(a) and obtain the same bound $\Omega(D\Delta)$.

Finally, for model~(c), it is sufficient to increase the range of each node $w_j$
in the last shifted copy of the network in $\m{F}$ in such a way that $s$ is in its range area.
}
\end{proof}

As we argue next,
the complexity of broadcasting depends also on granularity of the network.

\begin{corollary}\labell{cor:lowerGranularity}
For any deterministic broadcasting algorithm $\cA$ in
unknown uniform model,
and for any each $D\geq 3$ and $g\ge 4$, there exists a network with eccentricity
$D$ and granularity $g$ on which algorithm $\cA$ completes broadcasting in $\Omega(D g)$
rounds.
\end{corollary}
\begin{proof}
Note that granularity of the family of networks considered in the proof of
Theorem~\ref{th:lower} is $\Omega(\Delta)$, which immediately gives the claimed
result.
\end{proof}

\tj{
Finally, we make an observation that one can transform lower bounds from
Theorems~\ref{t:lower:log} and \ref{th:lower} to the case of randomized algorithms.
We sketch an idea of these transformations by considering networks from the family $\m{F}$
described in Theorem~\ref{th:lower}. Recall that each element of the layer $L_2$ should
transmit as the only element of $L_2$ in order to guarantee that the only element of $L_3$
is informed, regardless of its location. However, by simple counting arguments, the expectation
of the number of steps after which some of elements of $L_2$ transmit as the only ones is
$\Omega(\Delta)$.
}


\section{Conclusions}

In this work we provided several novel algorithmic techniques for broadcasting in ad hoc wireless
networks with uniform power,
supported by theoretical analysis. We also discovered that the lack of knowledge
about stations on close proximity results in substantially higher performance cost
for majority of network parameters $D,\Delta$, and even randomization does not help much.
The main open problem is to extend this study to networks with non-uniform power
and to other fundamental communication problems.


\begin{thebibliography}{10}

\bibitem{AvinEKLPR09}
C.~Avin, Y.~Emek, E.~Kantor, Z.~Lotker, D.~Peleg, and L.~Roditty.
\newblock Sinr diagrams: towards algorithmically usable sinr models of wireless
  networks.
\newblock In Tirthapura and Alvisi \cite{DBLP:conf/podc/2009}, pages 200--209.

\bibitem{AvinLPP09}
C.~Avin, Z.~Lotker, F.~Pasquale, and Y.~A. Pignolet.
\newblock A note on uniform power connectivity in the sinr model.
\newblock In S.~Dolev, editor, {\em ALGOSENSORS}, volume 5804 of {\em Lecture
  Notes in Computer Science}, pages 116--127. Springer, 2009.

\bibitem{Censor-HillelGKLN11}
K.~Censor-Hillel, S.~Gilbert, F.~Kuhn, N.~A. Lynch, and C.~C. Newport.
\newblock Structuring unreliable radio networks.
\newblock In C.~Gavoille and P.~Fraigniaud, editors, {\em PODC}, pages 79--88.
  ACM, 2011.

\bibitem{ClementiMS01}
A.~E.~F. Clementi, A.~Monti, and R.~Silvestri.
\newblock Selective families, superimposed codes, and broadcasting on unknown
  radio networks.
\newblock In S.~R. Kosaraju, editor, {\em SODA}, pages 709--718. ACM/SIAM,
  2001.

\bibitem{CzumajRytter-FOCS-03}
A.~Czumaj and W.~Rytter.
\newblock Broadcasting algorithms in radio networks with unknown topology.
\newblock In {\em FOCS}, pages 492--501. IEEE Computer Society, 2003.

\bibitem{DeMarco-SICOMP-10}
G.~DeMarco.
\newblock Distributed broadcast in unknown radio networks.
\newblock {\em SIAM J. Comput.}, 39(6):2162--2175, 2010.

\bibitem{DessmarkP07}
A.~Dessmark and A.~Pelc.
\newblock Broadcasting in geometric radio networks.
\newblock {\em J. Discrete Algorithms}, 5(1):187--201, 2007.

\bibitem{EmekGKPPS09}
Y.~Emek, L.~Gasieniec, E.~Kantor, A.~Pelc, D.~Peleg, and C.~Su.
\newblock Broadcasting in udg radio networks with unknown topology.
\newblock {\em Distributed Computing}, 21(5):331--351, 2009.

\bibitem{EmekKP08}
Y.~Emek, E.~Kantor, and D.~Peleg.
\newblock On the effect of the deployment setting on broadcasting in euclidean
  radio networks.
\newblock In R.~A. Bazzi and B.~Patt-Shamir, editors, {\em PODC}, pages
  223--232. ACM, 2008.

\bibitem{FanghanelKRV09}
A.~Fangh{\"a}nel, T.~Kesselheim, H.~R{\"a}cke, and B.~V{\"o}cking.
\newblock Oblivious interference scheduling.
\newblock In Tirthapura and Alvisi \cite{DBLP:conf/podc/2009}, pages 220--229.

\bibitem{Farach-ColtonAMMZ11}
M.~Farach-Colton, A.~F. Anta, A.~Milani, M.~A. Mosteiro, and S.~Zaks.
\newblock Brief announcement: Opportunistic information dissemination in mobile
  ad-hoc networks: - adaptiveness vs. obliviousness and randomization vs.
  determinism.
\newblock In D.~Peleg, editor, {\em DISC}, volume 6950 of {\em Lecture Notes in
  Computer Science}, pages 202--204. Springer, 2011.

\bibitem{Farach-ColtonM07}
M.~Farach-Colton and M.~A. Mosteiro.
\newblock Sensor network gossiping or how to break the broadcast lower bound.
\newblock In T.~Tokuyama, editor, {\em ISAAC}, volume 4835 of {\em Lecture
  Notes in Computer Science}, pages 232--243. Springer, 2007.

\bibitem{GalcikGL09}
F.~Galc\'{\i}k, L.~Gasieniec, and A.~Lingas.
\newblock Efficient broadcasting in known topology radio networks with
  long-range interference.
\newblock In Tirthapura and Alvisi \cite{DBLP:conf/podc/2009}, pages 230--239.

\bibitem{GasieniecKKPS08}
L.~Gasieniec, E.~Kantor, D.~R. Kowalski, D.~Peleg, and C.~Su.
\newblock Time efficient k-shot broadcasting in known topology radio networks.
\newblock {\em Distributed Computing}, 21(2):117--127, 2008.

\bibitem{GasieniecKLW08}
L.~Gasieniec, D.~R. Kowalski, A.~Lingas, and M.~Wahlen.
\newblock Efficient broadcasting in known geometric radio networks with
  non-uniform ranges.
\newblock In G.~Taubenfeld, editor, {\em DISC}, volume 5218 of {\em Lecture
  Notes in Computer Science}, pages 274--288. Springer, 2008.

\bibitem{GoussevskaiaMW08}
O.~Goussevskaia, T.~Moscibroda, and R.~Wattenhofer.
\newblock Local broadcasting in the physical interference model.
\newblock In M.~Segal and A.~Kesselman, editors, {\em DIALM-POMC}, pages
  35--44. ACM, 2008.

\bibitem{WatSurv}
O.~Goussevskaia, Y.~A. Pignolet, and R.~Wattenhofer.
\newblock Efficiency of wireless networks: Approximation algorithms for the
  physical interference model.
\newblock {\em Foundations and Trends in Networking}, 4(3):313--420, 2010.

\bibitem{Kesselheim11}
T.~Kesselheim.
\newblock A constant-factor approximation for wireless capacity maximization
  with power control in the sinr model.
\newblock In D.~Randall, editor, {\em SODA}, pages 1549--1559. SIAM, 2011.

\bibitem{KV10}
T.~Kesselheim and B.~V{\"o}cking.
\newblock Distributed contention resolution in wireless networks.
\newblock In N.~A. Lynch and A.~A. Shvartsman, editors, {\em DISC}, volume 6343
  of {\em Lecture Notes in Computer Science}, pages 163--178. Springer, 2010.

\bibitem{Kow-PODC-05}
D.~R. Kowalski.
\newblock On selection problem in radio networks.
\newblock In M.~K. Aguilera and J.~Aspnes, editors, {\em PODC}, pages 158--166.
  ACM, 2005.

\bibitem{KP04}
D.~R. Kowalski and A.~Pelc.
\newblock Time of deterministic broadcasting in radio networks with local
  knowledge.
\newblock {\em SIAM J. Comput.}, 33(4):870--891, 2004.

\bibitem{KP-DC-05}
D.~R. Kowalski and A.~Pelc.
\newblock Broadcasting in undirected ad hoc radio networks.
\newblock {\em Distributed Computing}, 18(1):43--57, 2005.

\bibitem{KP-DC-07}
D.~R. Kowalski and A.~Pelc.
\newblock Optimal deterministic broadcasting in known topology radio networks.
\newblock {\em Distributed Computing}, 19(3):185--195, 2007.

\bibitem{RichaSSZ}
A.~Richa, C.~Scheideler, S.~Schmid, and J.~Zhang.
\newblock Towards jamming-resistant and competitive medium access in the sinr
  model.
\newblock In {\em Proceedings of the 3rd ACM workshop on Wireless of the
  students, by the students, for the students}, S3 '11, pages 33--36, New York,
  NY, USA, 2011. ACM.

\bibitem{SenH96}
A.~Sen and M.~L. Huson.
\newblock A new model for scheduling packet radio networks.
\newblock In {\em INFOCOM}, pages 1116--1124, 1996.

\bibitem{DBLP:conf/podc/2009}
S.~Tirthapura and L.~Alvisi, editors.
\newblock {\em Proceedings of the 28th Annual ACM Symposium on Principles of
  Distributed Computing, PODC 2009, Calgary, Alberta, Canada, August 10-12,
  2009}. ACM, 2009.

\bibitem{YuWHL11}
D.~Yu, Y.~Wang, Q.-S. Hua, and F.~C.~M. Lau.
\newblock Distributed local broadcasting algorithms in the physical
  interference model.
\newblock In {\em DCOSS}, pages 1--8. IEEE, 2011.

\end{thebibliography}

\end{document}